\documentclass[journal,twoside,web]{ieeecolor}

\usepackage{generic}
\usepackage{amsmath,amssymb,amsfonts}
\usepackage{mathtools}
\usepackage{steinmetz}
\usepackage{algorithmic}
\usepackage{graphicx}
\usepackage{algorithm,algorithmic}
\usepackage{hyperref}
\usepackage{subcaption}
\hypersetup{
    colorlinks=true,
    linkcolor=blue,
    citecolor=blue,
    filecolor=magenta,      
    urlcolor=cyan,
    pdftitle={Overleaf Example},
    pdfpagemode=FullScreen,
    }
\usepackage{cite}
\usepackage{textcomp}
\usepackage{comment}
\newtheorem{theorem}{Theorem}
\newtheorem{lemma}{Lemma}
\newtheorem{definition}{Definition}
\newtheorem{assumption}{Assumption}
\newtheorem{remark}{Remark}
\newtheorem{corollary}{Corollary}
\def\BibTeX{{\rm B\kern-.05em{\sc i\kern-.025em b}\kern-.08em
    T\kern-.1667em\lower.7ex\hbox{E}\kern-.125emX}}
\makeatletter
\let\NAT@parse\undefined
\makeatother
\begin{document}
\title{Frequency Domain Stability and Convergence Analysis for General Reset Control Systems Architecture}

\author{S. Ali Hosseini, and S. Hassan HosseinNia, \IEEEmembership{Senior Member, IEEE}
\thanks{This work was supported by
ASMPT, 6641 TL Beuningen, The Netherlands.}
\thanks{S. Ali Hosseini and S. Hassan HosseinNia are with the Department of Precision and Microsystems
Engineering, Delft University of Technology, 2628 CD Delft, The Netherlands (e-mail: S.A.Hosseini@tudelft.nl; S.H.HosseinNiaKani@tudelft.nl).}
}
\maketitle

\begin{abstract}
A key factor that generates significant interest in reset control systems, especially within industrial contexts, is their potential to be designed using a frequency-domain loop-shaping procedure. On the other hand, formulating and assessing stability analysis for these nonlinear elements often depends on access to parametric models and numerically solving linear matrix inequalities. These specific factors could present challenges to the successful implementation of reset control within industrial settings.
Moreover, one of the most effective structures for implementing reset elements is to use them in parallel with a linear element. Therefore, this article presents the development of the frequency domain-based $H_\beta$ stability method from a series to a more general structure of reset control systems. Additionally, it investigates the behavior of different reset elements in terms of the feasibility of stability in the presence of time delay.
To illustrate the research findings, two examples are provided, including one from an industrial application.
\end{abstract}


\section{Introduction} \label{Introduction}
    \IEEEPARstart{R}{eset} elements are nonlinear filters used to overcome the fundamental performance limitations of linear time-invariant (LTI) control systems \cite{zhao2019overcoming}.
 The increasing demand for extremely fast and accurate performance in fields such as precision motion control is pushing linear controllers to their limits \cite{boyd1991linear}.
The concept of reset control was initially introduced in \cite{clegg1958nonlinear} as a nonlinear integrator, later known as the Clegg integrator (CI). It demonstrated promising behavior in overcoming the limitations inherent in linear feedback control caused by Bode’s gain-phase relationship \cite{CgLp}.
Over time, more advanced reset components were created, including the first-order reset element (FORE) \cite{horowitz1975non}, the generalized first-order reset element (GFORE) \cite{guo2009frequency}, and the second-order reset element (SORE) \cite{hazeleger2016second}.

By using the sinusoidal-input describing function (SIDF) method \cite{vander1968SIDF}, a reset element can be represented in the frequency domain. This allows us to perform frequency domain analysis when a reset element is incorporated with LTI elements in a control loop. Frequency response analysis is a technique that is used to assess the magnitude and phase properties of a control system. Thus, a designer can shape and tune the performance of closed-loop systems based on their open-loop analysis, a process known as loop shaping \cite{van2017frequency}. This design approach enables us to evaluate the closed-loop performance of the control system without developing a parametric model of the plant. Instead, we can use a frequency response function (FRF), which can be derived solely from measurement data \cite{franklin2002feedback}.

The stability analysis of a reset control system (RCS) is as important as the performance analysis to ensure reliable and predictable operation. To assess the stability of a closed-loop RCS, several methods have been proposed \cite{banos2010reset,paesa2011design,banos2012reset,guo2015analysis,vettori2014geometric,griggs2009interconnections,hollot1997stability}. These include the reset instants dependent method \cite{banos2010reset,paesa2011design}, the quadratic Lyapunov functions method \cite{banos2012reset,guo2015analysis,vettori2014geometric}, the small gain, passivity, and IQC approaches \cite{griggs2009interconnections,hollot1997stability}. Utilizing these approaches typically requires solving LMIs and deriving a parametric approximation of the plant, both of which are time-consuming processes and introduce uncertainties. Therefore, similar to frequency domain-based performance analysis, a frequency domain-based stability analysis is desired for RCSs.

To this respect, some frequency domain-based stability methods are introduced \cite{beker2004fundamental,beker1999stability,van2017frequency,van2024scaled,dastjerdi2023frequency,guo2015analysis}.
In \cite{van2017frequency}, a frequency domain-based stability method was developed for reset control systems, focusing on their input/output behavior. However, this method can be applied when the output is confined to a certain sector bound, as well as when the reset action is only triggered by the zero crossing of the reset element input.
In \cite{van2024scaled}, they approximate the scale graph \cite{huang2020tight} of reset controllers, resulting in a graphical tool to assess the stability of an RCS. However, it is not applicable to the zero-crossing triggered reset elements, which are the most common reset controllers.
In \cite{beker1999stability, beker2004fundamental}, the $H_\beta$ method was introduced, which includes a strictly positive real condition to guarantee closed-loop stability. It has been
proven that an RCS can be quadratically
stable if and only if the $H_\beta$ condition is satisfied. In fact,
the $H_\beta$ condition is a necessary and sufficient condition for the existence of a quadratic Lyapunov function \cite{beker2004fundamental,banos2010reset}. However, in general, this approach still relies on a parametric model to identify both a positive definite matrix and a vector that define the output of a transfer matrix, which must be strictly positive real.\\ Additionally, the method discussed in \cite{beker2004fundamental} is limited to RCSs where the error is forced to be the reset-triggered signal, and no pre-filtering of the reset element is possible. The same limitations apply to \cite{guo2015analysis}, as the quadratic stability presented there relies on the stability theorem from \cite{beker2004fundamental}. Although \cite[Remark 3.4]{guo2015analysis} introduces an additional condition in the $H_\beta$ method for non-zero after-reset values ($\gamma$), the exact proof and the effect of the shaping filter are not investigated there. In this study, we extend the $H_\beta$ method to accommodate any arbitrary reset surface (condition) and non-zero after-reset values ($\gamma$). While this extension is not the primary contribution of our work, it serves to validate the other theorems and lemmas presented. The proof is straightforward, drawing from existing studies such as \cite{banos2012reset}, \cite{guo2015analysis}, and \cite{beker2004fundamental}.

In \cite{dastjerdi2023frequency}, novel graphical stability conditions in the frequency domain are proposed for control systems incorporating first- and second-order reset elements, along with a shaping filter in the reset line. This approach facilitates the assessment of uniform bounded-input bounded-state (UBIBS) stability of reset control systems using the $H_\beta$ method without the need to solve LMI-based conditions. The matrix-based $H_\beta$ transfer function was first converted to an FRF-based transfer function in \cite[Corollary 3.12]{beker2001analysis} for a simple reset control system, where the reset element functions as the sole controller in the loop. However, in \cite{dastjerdi2023frequency}, a broader class of reset control systems is considered. Although an FRF-based transfer function is derived, the derivation is intuitive, lacking a formal mathematical connection between the matrix-based $H_\beta$ transfer function and the FRF-based transfer function for the proposed structure. Furthermore, the introduced method applies only to the series structure of a RCS without an LTI element in parallel with the reset element. Consequently, the analysis is not valid for reset systems with a feedthrough term. Therefore, since satisfying the $H_\beta$ conditions is a necessary requirement in the uniformly exponential convergence lemma presented in \cite{dastjerdi2022closed}, a significant gap remains in guaranteeing the convergence of non-series reset control systems using FRF-based methods.   

In \cite{banos2010reset} and \cite{guo2015analysis}, the effect of delay on the $H_\beta$ condition has been studied. However, the dynamics of the delay must still be known to calculate the $H_\beta$ transfer function. Based on the graphical method in \cite{dastjerdi2023frequency}, time delay could be part of the plant's FRF; thus, it directly affects the $H_\beta$ transfer function. In this study, we aim to investigate the effect of the $H_\beta$ transfer function with time delay on the quadratic stability conditions of a RCS when using different reset elements. This helps avoid wasting time assessing the stability of certain types of reset elements (GFORE, CI, PCI) using the FRF-based $H_\beta$ method in the presence of time delay.\\

Regarding the mentioned gaps in the existing literature, the contributions of this paper are as follows:
\begin{itemize}
    \item Provide the extended $H_\beta$ method for cases involving non-zero after-reset values, including the shaping filter.
    \item Mathematically calculate the FRF-based $H_\beta$ transfer function (in \cite{dastjerdi2023frequency} it is intuitively presented only for a limited class of RCSs):
    \begin{itemize}
        \item For general reset control systems, including parallel branch.
        \item Using matrix equality tools to prove the equivalence between a state-space-based transfer function and an FRF-based transfer function.
    \end{itemize}
    \item{Develop the conditions in the $H_\beta$ theorem for the new $H_\beta$ transfer function, while still requiring only the frequency response functions of the loop components.}
    \item Showing that the FRF-based $H_\beta$ method always leads to an infeasible stability assessment in the presence of delay when using CI or PCI.
\end{itemize}

The remainder of this paper is structured as follows. Section \ref{sec: preliminaries} provides descriptions of the reset element, LTI components, and the closed-loop system while covering some key theorems and lemmas that serve as the foundation for the rest of this study.
Section \ref{sec: main results} presents the main results of this paper. It first introduces the new FRF-based transfer function for the matrix-based $H_\beta$ transfer function and then develops the frequency domain-based conditions for the quadratic stability of general RCSs.
In Section \ref{sec: delay}, the effect of delay, which is an inseparable part of industrial settings, is examined in terms of its influence on the feasibility of this approach.
The utility and validation of the findings of this study are showcased through simulated examples in Section \ref{sec: example}. Finally, conclusions and suggestions for future studies are given in Section \ref{sec: conclusion}.
\section{Preliminaries}\label{sec: preliminaries}
\subsection{System description}
In the following, we provide a formal introduction to the reset control system and demonstrate its parallel interconnection with an LTI system, resulting in a SISO closed-loop system, as depicted in Fig. \ref{fig:Block diagram cl}. The linear part contains $G$ as the plant; $C_{1}$, $C_{2}$, and $C_{3}$ as the linear controllers; and a shaping filter $C_\text{s}$ with a proper stable transfer
function. Additionally, $C_\text{R}$ serves as the reset element, with its state-space realization presented as:
\begin{equation} 
		C_\text{R} : \begin{cases}
            \dot{x}_r(t)=A_rx_r(t)+B_ru_1(t),  \qquad \qquad e_r(t)\neq0,\\
            x_r(t^+)=\gamma x_r(t), \qquad e_r=0\wedge(1-\gamma)x_r(t)\neq0,\\
            u_r(t)=C_rx_r(t)+D_ru_1(t),
		\end{cases}
  \label{eq.SS reset}
	\end{equation}
where $A_r\in\mathbb{R}$, $B_r\in\mathbb{R}$, $C_r\in\mathbb{R}$, and $D_r\in\mathbb{R}$ represent the state-space matrices of the reset element, and the reset value is denoted by $\gamma \in\mathbb{R} $. $x_r(t)\in\mathbb{R}$ is the only state of the reset element, $x_r(t^{+})\in\mathbb{R}$ is the after reset state, $u_1(t)\in\mathbb{R}$ and $u_r(t)\in\mathbb{R}$ represent the input and output of the reset element, respectively. Also, $e_r(t)$ is the output of the shaping filter $C_s$.

\begin{figure}[!t]
\centerline{\includegraphics[scale=0.45,,trim=0 0 0 0,clip]{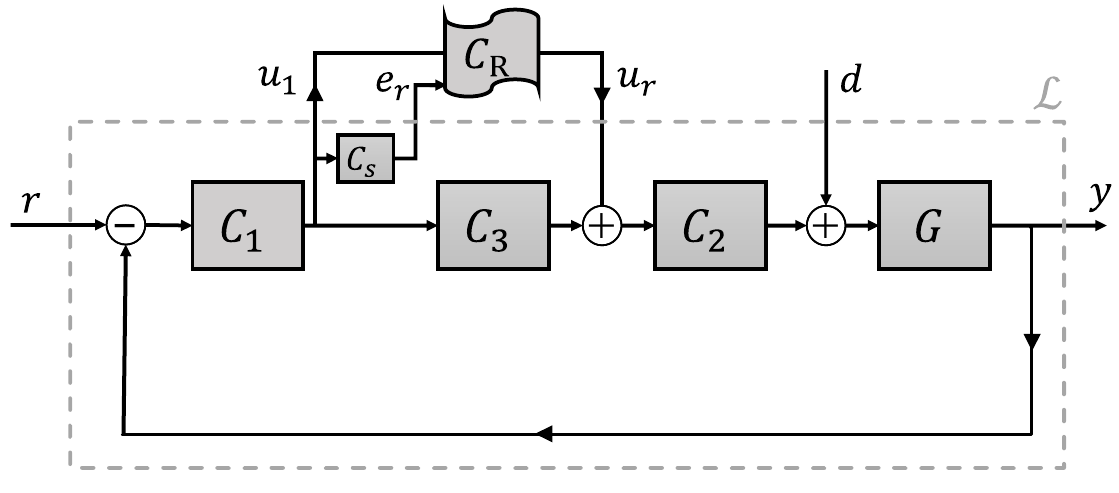}}
\caption{The closed-loop architecture of a general reset control system.}
\label{fig:Block diagram cl}
\end{figure}

In this study, $C_R$ is considered to be a first-order reset element (CI or GFORE). Due to the utilization of a parallel filter with the reset element, PCI is not considered, as it can be constructed using a CI in parallel with the desired gain. By considering $C_rB_r=\omega_k>0$ and $A_r=-\omega_r$ ($\omega_r\geq0$), the general form of the reset element used in this paper is defined, and its base linear transfer function is as follows
\begin{equation}
\begin{split}
    \label{RCS expression}
    R(s)&=C_r(s-A_r)^{-1}B_r+D_r \\
    &=\frac{\omega_k}{s+\omega_r}+D_r,
    \end{split}
\end{equation}
where $s \in\mathbb{C}$ is the Laplace variable. The reset element represents a CI or PCI when $\omega_r=0$ and represents a GFORE element when $\omega_r \neq0$.

For the LTI part of the closed-loop system, which is denoted by $\mathcal{L}$, we have
\begin{equation} 
		\mathcal{L} : \begin{cases}
			\label{eq.SS linear}
			\dot{x}_l(t)=Ax_l(t)+B_uu_r(t)+Bw(t),\\
			y(t)=Cx_l(t),\\
			u_1(t)=C_ux_l(t)+D_uw(t), \\
                e_r(t)=C_ex_l(t)+D_e w(t),
		\end{cases}
	\end{equation}
where $x_l(t)\in\mathbb{R}^{n_l}$ is the state of the LTI part of the system, and $w(t)=\begin{bmatrix}
    r(t) & d(t)
\end{bmatrix}^T\in\mathbb{R}^{2}$, in which $r(t)$ and $d(t)$ are the reference and disturbance signal of the system, respectively. Also, $A\in\mathbb{R}^{n_l \times n_l}$, $B\in\mathbb{R}^{n_l \times 2}$, $C\in\mathbb{R}^{1 \times n_l}$, $B_u\in\mathbb{R}^{n_l \times 1}$, $C_u\in\mathbb{R}^{1 \times n_l}$, $D_u\in\mathbb{R}^{1 \times 2}$, $C_e\in\mathbb{R}^{1 \times n_l}$ and $D_e\in\mathbb{R}^{1 \times 2}$ are the corresponding dynamic matrices.\\

\begin{assumption}
    \label{ass: assumption no direct feed through}
    In this study, it is assumed that there is no direct feedthrough from input $w(t)$ and $u_r(t)$ to plant output $y(t)$ and from $u_r(t)$ to $u_1(t)$ and $e_r(t)$.\\
\end{assumption}

Assumption \ref{ass: assumption no direct feed through} is reasonable, as it pertains to any causal LTI element $C_1$, $C_2$, $C_3$, and any plant \( G \) with a relative degree greater than zero, which includes the vast majority of motion and mass-based systems. Hence, the closed-loop state-space representation of the overall reset control system (RCS) can be expressed as follows
\begin{equation} 
		\text{RCS} : \begin{cases}
			\label{eq.SS closed loop}
            \dot{x}(t)=\bar{A}x(t)+\bar{B}w(t), \qquad   x(t)\notin \mathcal{F},\\
            x(t^+)=A_\rho x(t), \: \: \: \, \qquad   \qquad x(t)\in \mathcal{F},\\
            e_r(t)=\bar{C}_ex(t)+\bar{D}_e w(t), \\
            y(t)=\bar{C}x(t),
		\end{cases}
	\end{equation}
in which $x(t)=\begin{bmatrix}
    x_r(t)^T & x_l(t)^T
\end{bmatrix}^T \in\mathbb{R}^{1+n_l}$, $\bar{C}=\begin{bmatrix}
    0 & C
\end{bmatrix}$, $\bar{B}=\begin{bmatrix}  0_{1\times 2}\\ B \end{bmatrix} +$ $\begin{bmatrix}  B_rD_u & 0_{1\times 1}\\ B_uD_rD_u & 0_{n_l\times 1} \end{bmatrix}$, $\bar{A}=\begin{bmatrix}
    A_r & B_rC_u\\
    B_uC_r & A+B_uD_rC_u
\end{bmatrix}$, $A_\rho = \begin{bmatrix}
    \gamma & 0_{1\times n_l}\\
    
    0_{n_l\times 1} & I_{n_l\times n_l}
\end{bmatrix}$, and $\bar{C}_e=\begin{bmatrix}
    0 & C_e
\end{bmatrix}$. With reset surface $\mathcal{F}=\{x(t) \in \mathbb{R}^{n_l+1}:\bar{C}_ex(t)+\bar{D}_e w(t)=0 \wedge (I-A_\rho)x(t)\neq0\}$.

\subsection{Foundational theorems and lemmas}
This part covers the fundamental theorems and lemmas that form the foundation for the rest of this study.
As mentioned earlier, this study focuses on the $H_\beta$ stability approach. The $H_\beta$ condition is based on a quadratic Lyapunov function that must be decreasing over the entire state space along the system trajectories and non-increasing at the reset jumps. This method was first introduced in \cite{beker2004fundamental} for the case where the reset surface is defined as $\{x(t) \in \mathbb{R}^{n_l+1}:\bar{C}x(t)=0 \wedge (I-A_\rho)x(t)\neq0\}$ and $\gamma=0$, as in \cite[Chapter 4.4]{banos2012reset}. However, in this study, the reset surface is considered as $\mathcal{F}$, with non-zero cases for $\gamma$ also examined. Consequently, the modified $H_\beta$ method for assessing the stability of the reset control system in \eqref{eq.SS closed loop} is described in the following theorem.\\  

\begin{theorem}
 The zero equilibrium of the reset control system \eqref{eq.SS closed loop} with $w = 0$ is globally uniformly asymptotically stable if there exist $\varrho=\varrho^T > 0$ and $\beta\in\mathbb{R}$ such that the transfer function
\begin{equation}
    \label{eq.H beta}
    H_{\beta}(s)=C_0(sI-\bar{A})^{-1}B_0,
\end{equation}
with
\begin{equation}
 C_0=\begin{bmatrix}
\varrho & \beta C_e
\end{bmatrix}, \quad
B_0= \begin{bmatrix}
    1 \\
    0_{n_l \times 1}
\end{bmatrix},
\label{eq. C0B0}
\end{equation}
is Strictly Positive Real (SPR), $(\bar{A},B_0)$ and $(\bar{A},C_0)$ are controllable and observable respectively, and $-1\leq\gamma\leq1$.
\label{Th. Theorem1}
\end{theorem}
As previously mentioned, the only differences between Theorem \ref{Th. Theorem1} and the $H_\beta$ theorem presented in \cite[Proposition 4.5]{banos2012reset} are the reset surface and the value of $\gamma$. Therefore, the proof of Theorem \ref{Th. Theorem1} is provided in Appendix \ref{App: H_b}, following the approach used in \cite[Proposition 4.5]{banos2012reset}.\\


\begin{lemma}
\cite[Lemma 6.1]{khalil2002nonlinear}, Let $H(s)$ be a proper rational $p\times p$ transfer function and assume det$[H(s) + H^{T}(-s)]$ is not identically zero. Then, $H(s)$ is SPR if and only if :
\begin{itemize}
    \item $H(s)$ is Hurwitz,
    \item $H(j\omega) + H^{T}(-j\omega)$ is positive definite for all $\omega \in \mathbb{R}$,
    \item either $H(\infty) + H^{T}(\infty)$ is positive definite or
     if $H(\infty) + H^{T}(\infty)$ is positive semi-definite, \\ $\lim_{\omega \to \infty} \omega^2 M^{T}[H(j\omega) + H^{T}(-j\omega)]M>0$, for any $p (p - q)$ full rank matrix $M$ such that $M^{T}[H(\infty) + H^{T}(\infty)]M=0$, and $q = \text{rank}[H(\infty) + H^{T}(\infty)]$.\\
\end{itemize}
\label{lem: SPR}
\end{lemma}

\begin{definition}
    \label{well posedness}
\cite{dastjerdi2023frequency}, A time $\bar{T} > 0$ is called a reset instant for the reset control system \eqref{eq.SS closed loop} if $e_r(\bar{T})=0\,\wedge\,(I-A_\rho)x(\bar{T})\neq0$. For any given initial condition and input $w$, the resulting set of all reset instants defines the reset sequence $\{t_k\}$, with $t_k \leq t_{k+1}$ for all $k \in \mathbb{N}$. The reset instants $t_k$ have the well-posedness property if for any initial condition $x_0$ and any input $w$, all the reset instants are distinct, and there exists $\lambda > 0$ such that, for all $k \in \mathbb{N}$, $\lambda \leq t_{k+1} - t_k$.\\
\end{definition}
Note that the second condition in $\mathcal{F}$, $(I- A_\rho)x(t)\neq0$, is imposed for well-posedness, to avoid the so-called re-resetting, that is, the fact that immediately after a reset, the state vector could satisfy the reset condition again, implying formally an infinite sequence of re-settings to be performed
instantaneously, also referred to as beating. Thus, to avoid this, the term $(I- A_\rho)x(t)\neq0$ is added in the definition of $\mathcal{F}$ (see \cite[Section 1.4.1]{banos2012reset}). Simply it means
\begin{equation}
    \nonumber
    x(t)\in \mathcal{F} \Rightarrow  x(t^{+})\notin \mathcal{F}.\\
\end{equation}

\begin{assumption}
    \label{assumption}
    There are infinitely reset instants and $\lim_{k \to \infty} t_k=\infty$.\\
\end{assumption}

Assumption \ref{assumption} is necessary to guarantee the existence of a rest instant $t_k$ as $k\rightarrow \infty$ for the proof of Lemma \ref{lemma UBIBS} and Lemma \ref{convrgence} (will be presented further). However, if it does not hold, it means there is no reset after a time $t_{k'}$, which implies that the RCS behaves like its base linear system ($A_\rho=I$ in \eqref{eq.SS closed loop}). In this case, the stability and convergence properties of the RCS are equivalent to those of its base linear system. This means that if the base linear system is stable and has a convergent solution, then the RCS is stable and has a convergent solution as well.\\
In the following, the definition of the Bohl function and its characteristics are presented, as UBIBS stability and convergence of the RCSs have been established for this class of input signals in Lemma \ref{lemma UBIBS} and \ref{convrgence}.
\begin{definition}
    \label{def: bohl function}
    \cite[Definition 2.5]{Bohl_trentelman2002control}, A function that is a linear combination of functions of the form $t^k e^{\lambda t}$, where the $k$'s are nonnegative integers and $\lambda \in \mathbb{C}$, is called a Bohl function. The numbers $\lambda$ that appear in this linear combination (and cannot be canceled) are called the characteristic exponents of the Bohl function.\\
\end{definition}

\begin{remark}
    In \cite[Theorem 2.7]{Bohl_trentelman2002control} it is shown that if $p$ and $q$ are Bohl functions then $p+q$, $pq$ and $\dot{p}$ are also Bohl functions. Additionally, it can easily be shown that step functions, ramp functions, and any sinusoidal functions are Bohl functions.\\
\end{remark}

\begin{lemma}
    \label{lemma UBIBS}
    \cite[Lemma 2]{dastjerdi2023frequency}, Consider the reset control system \eqref{eq.SS closed loop}. Suppose that
    \begin{itemize}
        \item Assumption \ref{assumption} holds;
        \item $-1<\gamma<1$;
    \item All conditions in Theorem \ref{Th. Theorem1} hold;
        \item at least one of the following conditions holds:
        \begin{itemize}
            \item $C_s=1$;
            \item the reset instants have the well-posedness property.
        \end{itemize}
    \end{itemize}
    Then, the reset control system \eqref{eq.SS closed loop} has a well-defined unique left-continuous response for any initial condition $x_0$ and any input $w$, which is a Bohl function. In addition, the reset control system \eqref{eq.SS closed loop} has the UBIBS property.\\
\end{lemma}

Please note that the well-posedness property (Definition \ref{well posedness}) is a reasonable assumption, as it is ensured by the second condition ($(I-A_\rho)x(t) \neq 0$) on the reset surface $\mathcal{F}$. This condition effectively prevents what is known as 're-resetting,' a situation where, immediately after a reset, the state vector may once again satisfy the reset condition, potentially resulting in an infinite sequence of resets (see \cite[Sections 1.41 and 2.2.1]{banos2012reset}).\\

\begin{definition}
    \label{def: convergence}
    \cite[Definition 2]{pavlov2007frequency}, The system $\dot{x}(t) \in F(x(t),w(t))$, with a given continuous on $\mathbb{R}$ input $w(t)$ is said to be (uniformly, exponentially) convergent if:
    \begin{itemize}
        \item all solutions $x_w(t,t_0,x_0)$ are defined for all $t \in [t_0,+\infty)$ and all initial conditions $t_0 \in \mathbb{R}$, $x_0\in\mathbb{R}^n$;
        \item there is a solution $\bar{x}_w(t)$ defined and bounded on $\mathbb{R}$;
        \item the solution $\bar{x}_w(t)$ is (uniformly, exponentially) globally asymptotically stable. \\
    \end{itemize}
\end{definition}

\begin{lemma}
\label{convrgence}
\cite[Lemma 2]{dastjerdi2022closed}, Consider the reset control system in \eqref{eq.SS closed loop} with $C_u=C_e$ and $D_u=D_e$ ($C_s=1$), the system is uniformly exponentially convergent \cite[Definition 2]{pavlov2007frequency} for any input $w$ that qualifies as a Bohl function, if
\begin{itemize}
 \item Assumption \ref{assumption} holds;
        \item $-1<\gamma<1$;
    \item All conditions in Theorem \ref{Th. Theorem1} hold;
\item the initial condition of the reset element is zero.\\
\end{itemize}
\end{lemma}

\begin{remark}
    \label{proof explanation}
    In Lemma \ref{convrgence}, it is assumed $C_u=C_e$ and $D_u=D_e$ since \cite{dastjerdi2022closed} only considers systems with $C_s=1$.\\
\end{remark}

The main stability condition presented in Theorem \ref{Th. Theorem1} is in the frequency domain; however, it still requires the system model parameters to calculate $H_\beta$ transfer function and the use of an LMI-based method to find $\varrho$ and $\beta$. In \cite[Theorem 2]{dastjerdi2023frequency}, the results from Theorem \ref{Th. Theorem1} and Lemma \ref{lemma UBIBS} are combined into FRF-based conditions for a class of RCSs where $C_3=0$ and $D_r=0$ (the result presented in \cite{dastjerdi2023frequency} is not valid for $D_r\neq0$) using an intuitive derivation of an FRF-based $H_\beta$ transfer function.\\

The following two lemmas are presented as the foundation for calculating the FRF-based $H_\beta$ transfer function for the RCS \eqref{eq.SS closed loop} in the next section.
\begin{lemma}[Partitioned matrix inversion] 
    \label{lem: Mat inv}
     \quad \\
    \cite[Section 9.2.14]{mahmoud2021cyberphysical}, Let's consider an invertible matrix $M$, composed of blocks $Q_1$, $Q_2$, $Q_3$, and $Q_4$, as follows
    \begin{equation}
    \label{eq: M}
        M = \begin{bmatrix}
    Q_1 & Q_2\\
    Q_3 & Q_4
\end{bmatrix},
    \end{equation}
    the inverse of matrix $M$ can be expressed using a similarly structured block representation
    \begin{equation}
        M^{-1} =
\begin{bmatrix}
W & X \\
Y & Z \\
\end{bmatrix},
    \end{equation}
    where the blocks $W$, $X$, $Y$, $Z$, are as follows
\begin{equation}
\begin{cases}
W = \left( Q_1 - Q_2 Q_4^{-1} Q_3 \right)^{-1}, \\
Z = \left( Q_4 - Q_3 Q_1^{-1} Q_2 \right)^{-1},
\end{cases}
\end{equation}
and,
\begin{equation}
\begin{cases}
Y = -Q_4^{-1} Q_3 W, \\
X = -Q_1^{-1} Q_2 Z.
\end{cases}
\end{equation}\\
\end{lemma}

\begin{lemma}[Woodbury matrix identity]
    \label{lem: Woodbury}
    \quad \\
      \cite[Appendix. A, Solution to Problem 13.9]{higham2002accuracy}, For matrices $K$, $U$, $J$, and $V$, where $K$ is a $n\times n$ matrix, $J$ is a $k\times k$ matrix, $U$ is a $n\times k$ matrix, and $V$ is a $k\times n$ matrix, the following equation holds
    \begin{equation}
    \begin{split}
        (K + &UJV)^{-1} \\ &= K^{-1} - K^{-1}U(J^{-1} + VK^{-1}U)^{-1}VK^{-1}.
        \end{split}
    \end{equation}
\end{lemma}


\section{FRF based $H_\beta$ method for parallel reset control systems } \label{sec: main results}
In this section, we present our main result in the form of a theorem. Our objective is to employ the $H_\beta$ method in order to demonstrate the stability of the system described in (\ref{eq.SS closed loop}) in the frequency domain. We aim to transform the conditions specified by the $H_\beta$ method in such a manner that allows us to fulfill them by utilizing the measured FRF instead of the state-space model of the plant.
Regarding the transfer function in (\ref{eq.H beta}), it can be observed that it is not possible to directly utilize the measured FRF to assess the stability mentioned in Theorem \ref{Th. Theorem1} and Lemma \ref{lemma UBIBS}, because matrix $\bar{A}$ should be known. To this respect, in Lemma \ref{lem: Hb transfer}, the $H_\beta$ transfer function is rewritten.\\

\begin{lemma}
\label{lem: Hb transfer}
    Under Assumption \ref{ass: assumption no direct feed through}, the transfer function in (\ref{eq.H beta}), can be rewritten as follows
    \begin{equation}
    \label{eq: H_b FRF}
    \resizebox{1\hsize}{!}{
    $H_\beta(s)=\frac{\beta^{'}L(s)C_s(s)\Big(R(s)-D_r\Big)+\varrho^{'}\bigg(1+L(s)\Big(C_3(s)+D_r\Big)\bigg)\Bigl(R(s)-D_r\Bigl)}{1+L(s)\Bigl(R(s)+C_3(s)\Bigl)},$
    }
    \end{equation}
    where $L(s)=C_1(s)C_2(s)G(s)$, $\beta^{'}=\frac{-\beta}{B_r}$, and $\varrho^{'}=\frac{\varrho}{C_r B_r}$, which $C_r \in \mathbb{R}$ and $B_r \in \mathbb{R}$ are the one-dimensional input and output matrix of the reset element $C_R$ with $B_r C_r >0$.\\
    \textbf{ Proof.} Appendix \ref{App: I}\\
\end{lemma}

Regarding the Theorem \ref{Th. Theorem1} it is crucial to determine the sign of the real part of $H_\beta(s)$ to prove SPRness. Thus, the Nyquist Stability Vector (NSV) is given by the following definition.\\

\begin{definition}
\label{def:NSV}
    For the transfer function in (\ref{eq: H_b FRF}), the NSV is defined as below ($\forall \,\omega \in \mathbb{R}$)
    \begin{equation}
    \begin{split}
    \stackrel{\rightarrow}{\mathcal{N}}(\omega)=&\begin{bmatrix}
        \mathcal{N}_x (\omega) & \mathcal{N}_y (\omega)
    \end{bmatrix}^T \\
    = &\begin{bmatrix}
        \mathfrak{R}(M_1^{*}(j\omega)M_2(j\omega)) & \mathfrak{R}(M_1^{*}(j\omega)M_3(j\omega))
    \end{bmatrix}^T,
    \label{eq.NSV}
    \end{split}
\end{equation}
where $ \mathfrak{R}(.)$ means the real part, $(.)^*$ means the complex conjugate, $j=\sqrt{-1}$, and
\begin{equation}
    \begin{split}
        M_1(j\omega)&=1+L(j\omega)\Bigl(R(j\omega)+C_3(j\omega)\Bigl), \\
        M_2(j\omega)&=L(j\omega)C_s(j\omega)\Big(R(j\omega)-D_r\Big), \\
        M_3(j\omega)&=\bigg(1+L(j\omega)\Big(C_3(j\omega)+D_r\Big)\bigg)\Bigl(R(j\omega)-D_r\Bigl).
    \end{split}
    \label{eq.M}
\end{equation}\\
\end{definition}

As previously stated, we aim to employ FRF-based conditions instead of matrix-based LMIs. Therefore, the following definitions are introduced as a foundation for the main theorem.
\begin{definition}
    \label{def:Types}
     We define $\theta_{\mathcal{N}}(\omega)=\phase{\overrightarrow{\mathcal{N}}(\omega)}$ where $\phase{\overrightarrow{\mathcal{N}}(\omega)}=\mathrm{tan}^{-1}\big({\frac{\mathcal{N}_y (\omega)}{\mathcal{N}_x (\omega)}}\big) \, \forall \,\omega \in \mathbb{R}$, and also
    \begin{equation}
\theta_1=\smash{\displaystyle\min_{\forall\omega \in\mathbb{R}^{}}}\theta_{\mathcal{N}}(\omega),  \quad \theta_2=\smash{\displaystyle\max_{\forall\omega \in\mathbb{R}^{}}}\theta_{\mathcal{N}}(\omega).
\label{eq: teta}
    \end{equation}
    In this paper for simplicity it is considered $\theta_{\mathcal{N}}(\omega)\in [-\frac{\pi}{2},  \frac{3\pi}{2})$.\\
\end{definition}

\begin{definition}
    \label{eq: LCs}
    The transfer functions $L(s)C_s(s)$, is defined as:
\begin{equation}
\label{eq: L(s)Cs}
     L(s)C_s(s)=\frac{K_{m} s^{m} + K_{m-1} s^{m-1} + \ldots + K_{m_0}}{K_{n} s^{n} + K_{n-1} s^{n-1} + \ldots + K_{n_0}}.
\end{equation}\\
\end{definition}
To be able to use Lemma \ref{lemma UBIBS} and Lemma \ref{convrgence} in the frequency domain, it is necessary to convert the conditions stated in Theorem \ref{Th. Theorem1} into frequency domain-based conditions. Therefore, with the presented lemmas and definitions, Theorem \ref{lem:Hb frequency} is introduced to establish the same property as stated in Theorem \ref{Th. Theorem1}, with conditions exclusively in the frequency domain.\\

\begin{theorem}
\label{lem:Hb frequency}
   The zero equilibrium of the reset control system (\ref{eq.SS closed loop}) satisfying the Assumption \ref{ass: assumption no direct feed through}, with $w(t) = 0$ is globally uniformly asymptotically stable if all the conditions listed below are satisfied
   \begin{itemize}
    \item The base linear system is stable, and its open loop transfer function does not have any unstable pole-zero cancellation.
    \item The shaping filter $C_s(s)$ is proper and stable.
    \item $-1<\gamma<1$.
    \item $B_r C_r >0$.
    \item $(\theta_2-\theta_1)<\pi$.
     \end{itemize}
     In the case of $\omega_r \neq 0$ (GFORE)
    \begin{itemize}
        \item $-\frac{\pi}{2}<\theta_{\mathcal{N}} (\omega)<\pi$ \text{and/or} $0<\theta_{\mathcal{N}} (\omega)<\frac{3\pi}{2}$, \quad for all $\omega \in [0,\infty)$.
    \end{itemize}
    In the case of $\omega_r=0$ (CI)
    \begin{itemize}
        \item The relative degree of the transfer function $L(s)$ must be 1.
         \item If $\lim_{s \to \infty} \operatorname{phase}\left(L(s) C_s(s)\right) = -90$ ($\frac{K_{n}}{K_{m}} > 0$), then $0 < \theta_{\mathcal{N}}(\omega) < \frac{3\pi}{2}$ for all $\omega \in [0,\infty)$.
    \item If $\lim_{s \to \infty} \operatorname{phase}\left(L(s) C_s(s)\right) = -270$ ($\frac{K_{n}}{K_{m}} < 0$), then $-\frac{\pi}{2} < \theta_{\mathcal{N}}(\omega) < \pi$ for all $\omega \in [0,\infty)$.
   \end{itemize}
   \textit{\textbf{Proof.}} Appendix \ref{App: II}\\
\end{theorem}

\begin{corollary}
\label{col: new convergence}
Based on Lemma \ref{convrgence}, suppose that $w$ is a Bohl function and Theorem \ref{lem:Hb frequency} holds. Then, the reset control system \eqref{eq.SS closed loop} with zero initial condition and $C_u=C_e$ and $D_u=D_e$ is uniformly exponentially convergent.\\
\end{corollary}

Suppose that the base linear system is stable and $-1 < \gamma < 1$. Then, according to Theorem \ref{lem:Hb frequency}, to investigate the stability of a reset control system, it is only necessary to plot $\theta_{\mathcal{N}}(\omega)$ and check the conditions associated with it. This process only requires the measured FRF of the plant. In this regard, with respect to Lemma \ref{lemma UBIBS} and Lemma \ref{convrgence}, the UBIBS and uniformly exponentially convergent property of RCS in \eqref{eq.SS closed loop} can also be achieved in the frequency domain by satisfying the conditions in Theorem \ref{lem:Hb frequency} rather than those in Theorem \ref{Th. Theorem1}.

Furthermore, it can be observed that the final condition in the case of $\omega_r=0$ is overly conservative, requiring the sum of the relative degrees of the plant, pre-filter, and post-filter to be 1. This conservatism makes the use of this method for reset controllers with CI nearly impossible. In other words, designing the controller based on the GFORE element is the preferred choice, as it offers advantages in terms of both performance \cite{LukeIFAC2024} and stability.\\

\begin{remark}
    \label{delay_exp}
To analytically analyze the effect of delay on the stability conditions in Theorem \ref{lem:Hb frequency} while using different reset elements, we approximate the system delay using the Pade approximation. This is done with a sufficient number of stable poles to ensure a Hurwitz $H_\beta$ transfer function based on Lemma \ref{lem: SPR}. This approximation is then integrated into the dynamics of the linear components of the system, as described in \eqref{eq.SS linear}.
\end{remark}
\section{Effect of delay} \label{sec: delay}
Regarding the Definition \ref{def:NSV}, the system's plant ($G$) is a component of both $\mathcal{N}_x (\omega)$ and $\mathcal{N}_y (\omega)$ in the NSV. Considering a plant in the presence of a delay can bring about changes in $\theta_{\mathcal{N}} (\omega)$ and consequently in the stability status. Thus, in this section, it is investigated in which circumstances the delay could lead to a situation where demonstrating the stability of the closed-loop reset control system in \eqref{eq.SS closed loop} becomes unfeasible. The findings of this analysis are summarized in Corollary \ref{CI,GFORE}. To support the proof of this corollary, we introduce Definition \ref{Zl}, Lemma \ref{lemma: sum_mult Zl}, and Lemma \ref{NxNyZl} as supplementary tools.\\

\begin{definition}
    \label{Zl}
    The function $K(x)$ is called a $\mathcal{Z}_L$ function if it can be defined in the form $K(x)=g(x)h(x)$, where
    \begin{equation}
        \label{ZL function}
         \lim_{x\to\infty} |g(x)|=0,
    \end{equation}
    and $h(x)$ be an oscillating function with zero mean as follows
    \begin{equation}
        \label{ZL function h}
         h(x)=A\sin{(x+\phi)},
    \end{equation}
    with $A\in \mathbb{R}^{+}$ and $\phi \in [0,2\pi).$\\
\end{definition}

\begin{lemma}
    \label{lemma: sum_mult Zl}
    If $K_1(x_1)$ and $K_2(x_2)$ ($K_1(x_1)\neq K_2(x_2)$) are $\mathcal{Z}_L$ functions then $K_1(x_1)+K_2(x_2)$ and $K_1(x_1)\times K_2(x_2)$ are also $\mathcal{Z}_L$ functions.\\
    \textit{\textbf{Proof.}} Appendix \ref{App: III}\\
\end{lemma}

\begin{lemma}
    \label{NxNyZl}
    In the presence of the delay, if both $\mathcal{N}_x (\omega)$ and $\mathcal{N}_y (\omega)$ in the NSV are $\mathcal{Z}_L$ functions, it is not possible to demonstrate the stability of the reset control system \eqref{eq.SS closed loop} using Theorem \ref{lem:Hb frequency}. 
\end{lemma}
\begin{proof}
    \label{P.NxNyZl}
    When both $\mathcal{N}_x (\omega)$ and $\mathcal{N}_y (\omega)$ are $\mathcal{Z}_L$ functions, it means they take both positive and negative values as $\omega \rightarrow \infty$, which implies that $\stackrel{\rightarrow}{\mathcal{N}}(\omega)=\begin{bmatrix}
        \mathcal{N}_x (\omega) & \mathcal{N}_y (\omega)
    \end{bmatrix}^T$ spans all four quadrants. Therefore, $\theta_{\mathcal{N}} (\omega)$ takes all values in $(0,2\pi]$, and non of the constraints on $\theta_{\mathcal{N}} (\omega)$ in Theorem \ref{lem:Hb frequency} cannot be satisfied.\\
\end{proof}

\begin{corollary}
    \label{CI,GFORE}
    For a system with delay, Theorem \ref{lem:Hb frequency} cannot be employed to establish the stability of the RCS in \eqref{eq.SS closed loop} with CI ($\omega_r=0$). However, it can be utilized to demonstrate the stability of a reset control system with a GFORE element ($\omega_r\neq0$).
\end{corollary}
\begin{proof}
    We will demonstrate that in the case of CI, both $\mathcal{N}_x (\omega)$ and $\mathcal{N}_y (\omega)$ are $\mathcal{Z}_L$ functions but in the case of GFORE, only $\mathcal{N}_x (\omega)$ is a $\mathcal{Z}_L$ function. Therefore, using Lemma \ref{NxNyZl}, the conditions in Theorem \ref{lem:Hb frequency} cannot be satisfied if the system contains a CI element.\\
Here we consider the delay function as 
\begin{equation}
\label{Euler}
    D(s,T)=e^{-Ts},
\end{equation}
with a delay time of $T\in  \mathbb{R}_{>0}$. Also, we can write 
\begin{equation}
    D(\omega,T)=\cos{(\omega T)}-j\sin{(\omega T)}.
\end{equation}
Thus, the plant in the presence of the delay can be expressed as follows:
\begin{equation}
\label{Gd}
G_d(j\omega)=G(j\omega)D(\omega,T),
\end{equation}
and, subsequently,
\begin{equation}
\label{Ld}
    L_d(j\omega)=L(j\omega)D(\omega,T).
\end{equation}
Therefore, equations in (\ref{eq.M}) can be rewritten as below
\begin{equation}
    \begin{split}
        M_1(j\omega)&=1+\bigg(L(j\omega) D(\omega,T)\Bigl(R(j\omega)+C_3(j\omega)\Bigl)\bigg), \\
        M_2(j\omega)&=L(j\omega) D(\omega,T)C_s(j\omega)\Big(R(j\omega)-D_r\Big), \\
        M_3(j\omega)&=\bigg(1+L(j\omega) D(\omega,T)\Big(C_3(j\omega)+D_r\Big)\bigg)...\\
        &\Bigl(R(j\omega)-D_r\Bigl),
    \end{split}
    \label{eq.M new}
\end{equation}
where based on the Definition \ref{def:NSV}, in order to derive $\mathcal{N}_x (\omega)$ and $\mathcal{N}_y (\omega)$, it is needed to calculate $M_1^*(j\omega)$, $M_2(j\omega)$, and $M_3(j\omega)$.\\
To calculate $M_1^*(j\omega)$, consider
\begin{equation}
\label{eq: a1b1}
    L(j\omega)\Bigl(R(j\omega)+C_3(j\omega)\Bigl)=a_1(\omega)+jb_1(\omega),
\end{equation}
where because $L(j\omega)$ is a strictly proper transfer function, $\lim_{\omega\to\infty} |a_1(\omega)|=0$, and $\lim_{\omega\to\infty} |b_1(\omega)|=0$. Hence,
\begin{equation}
\begin{split}
\label{eq: E1E2}
    &L(j\omega)D(\omega,T)\Bigl(R(j\omega)+C_3(j\omega)\Bigl)\\
    &=\Big(a_1(\omega)+jb_1(\omega)\Big)\Big(\cos{(\omega T)}-j\sin{(\omega T)}\Big) \\
    &=\Big(a_1(\omega)\cos{(\omega T)}+b_1(\omega)\sin{(\omega T)}\Big)\\
    &+j\Big(-a_1(\omega)\sin{(\omega T)}+b_1(\omega)\cos{(\omega T)}\Big)\\
    &=E_1(\omega)+jF_1(\omega),
    \end{split}
\end{equation}
thus,
\begin{equation}
    \label{eq: M1*}
    M_1^{*}(j\omega)=1+E_1(\omega)+jF_1(\omega),
\end{equation}
where regarding the Definition \ref{Zl}, $E_1(\omega)$ and $F_1(\omega)$ are $\mathcal{Z}_L$ functions.\\
To calculate $M_2(j\omega)$ it is considered
\begin{equation}
    \label{eq: M2}
     L(j\omega)C_s(j\omega)=a_2(\omega)+jb_2(\omega),
\end{equation}
and from \eqref{RCS expression},
\begin{equation}
    \label{eq: M2'}
    \begin{split}
     R(j\omega)-D_r&=\frac{\omega_k}{j\omega+\omega_r} \\
     &=\frac{\omega_k \omega_r}{\omega^2+\omega_r ^{2}}-j\frac{\omega_k\omega}{\omega^2+\omega_r ^{2}}=a_3(\omega)-jb_3(\omega).
     \end{split}
\end{equation}
Therefore,
 \begin{equation}
    \label{eq: M2''}
    \begin{split}
    &L(j\omega)C_s(j\omega)\Big(R(j\omega)-D_r\Big)\\
     &=\Big(a_2(\omega)a_3(\omega)+b_2(\omega)b_3(\omega)\Big) \\
     &+j\Big(b_2(\omega)a_3(\omega)-b_3(\omega)a_2(\omega)\Big)=E_2(\omega)+jF_2(\omega),
     \end{split}
\end{equation}
where $\lim_{\omega\to\infty} |E_2(\omega)|=0$, and $\lim_{\omega\to\infty} |F_2(\omega)|=0$. Thus, for $M_2(j\omega)$ it gives,
\begin{equation}
    \label{eq: M2'''}
    \begin{split}
      M_2(j\omega)&=L(j\omega)C_s(j\omega)\Big(R(j\omega)-D_r\Big)D(\omega,T)\\
     &=\Big(E_2(\omega)+jF_2(\omega)\Big)\Big(\cos{(\omega T)}-j\sin{(\omega T)}\Big)\\
     &=\Big(E_2(\omega)\cos{(\omega T)}+F_2(\omega)\sin{(\omega T)}\Big)\\
    &+j\Big(F_2(\omega)\cos{(\omega T)}-E_2(\omega)\sin{(\omega T)}\Big)\\
    &=E_3(\omega)+jF_3(\omega),
     \end{split}
\end{equation}
where $E_3(\omega)$ and $F_3(\omega)$ are $\mathcal{Z}_L$ functions.

For $M_3(\omega)$, first similar to what has been calculated in (\ref{eq: a1b1}) and (\ref{eq: E1E2}), we have
 \begin{equation}
 \label{M3'}
     \begin{split}
         &1+L(j\omega) D(\omega,T)\Big(C_3(j\omega)+D_r\Big)\\
         &=1+E_4(\omega)+jF_4(\omega),
     \end{split}
 \end{equation}
 where $E_4(\omega)$ and $F_4(\omega)$ are also $\mathcal{Z}_L$ functions. Then, considering
 \begin{equation}
     \begin{split}
     \label{eq: M''''}
         M_3(j\omega)&=\bigg(1+L(j\omega) D(\omega,T)\Big(C_3(j\omega)+D_r\Big)\bigg)...\\
         &\Bigl(R(j\omega)-D_r\Bigl),
              \end{split}
 \end{equation}
 and replacing \eqref{eq: M2'} and \eqref{M3'} in \eqref{eq: M''''}, results in
 \begin{equation}
     \begin{split}
     \label{eq: M'''''}
         M_3(j\omega)&=\Big(1+E_4(\omega)+jF_4(\omega)\Big)\Big(a_3(\omega)-jb_3(\omega)\Big)\\
         &=\bigg(a_3(\omega)+E_4(\omega)a_3(\omega)+F_4(\omega
         )b_3(\omega)\bigg)\\
         &+j\bigg(-b_3(\omega)-E_4(\omega)b_3(\omega)+F_4(\omega) a_3(\omega)\bigg)\\
         &=\Big(a_3(\omega)+E_5(\omega)\Big)-j\Big(b_3(\omega)+F_5(\omega)\Big),
     \end{split}
 \end{equation}
where $E_5(\omega)$ and $F_5(\omega)$ are $\mathcal{Z}_L$ functions.

Now we are able to calculate $\mathcal{N}_x (\omega)$ and $\mathcal{N}_y (\omega)$. From (\ref{eq.NSV}), we have
\begin{equation}
    \mathcal{N}_x(\omega)=\mathfrak{R}(M_1^{*}(j\omega)M_2(j\omega)),
\end{equation}
where, 
\begin{equation}
\begin{split}
    &M_1^{*}(j\omega)M_2(j\omega)\\
    &=\Big(1+E_1(\omega)+jF_1(\omega)\Big)\Big(E_3(\omega)+jF_3(\omega)\Big)\\
    &=\bigg(E_3(\omega)+E_1(\omega)E_3(\omega)-F_1(\omega)F_3(\omega)\bigg)\\
         &+j\bigg(F_1(\omega)E_3(\omega)+F_3(\omega)+E_1(\omega)F_3(\omega)\bigg)\\
         &= E_6(\omega)+jF_6(\omega),
    \end{split}
    \end{equation}
which $E_6(\omega)$ and $F_6(\omega)$ are $\mathcal{Z}_L$ functions. Thus,
\begin{equation}
\begin{split}
    \mathcal{N}_x (\omega)=\mathfrak{R}\Big(E_6(\omega)+jF_6(\omega)\Big)=E_6(\omega),
    \end{split}
\end{equation}
then, it concludes that $\mathcal{N}_x (\omega)$ is always a $\mathcal{Z}_L$ function.

For $\mathcal{N}_y (\omega)$, we have
\begin{equation}
    \mathcal{N}_y (\omega)=\mathfrak{R}(M_1^{*}(j\omega)M_3(j\omega)),
\end{equation}
where, 
\begin{equation}
\begin{split}
    &M_1^{*}(j\omega)M_3(j\omega)\\
    &=\bigg(1+E_1(\omega)+jF_1(\omega)\bigg)\bigg(\Big(a_3(\omega)+E_5(\omega)\Big)\\
    &-j\Big(b_3(\omega)+F_5(\omega)\Big)\bigg)\\
    &=\Big(a_3(\omega)+E_5(\omega)+E_1(\omega)a_3(\omega)+E_1(\omega)E_5(\omega)\\
   &+F_1(\omega)b_3(\omega)+F_1(\omega)F_5(\omega)\Big)+j\Big(F_1(\omega)a_3(\omega)\\
   &+F_1(\omega)E_5(\omega)-b_3(\omega)-F_5(\omega)-E_1(\omega)b_3(\omega)\\
   &-E_1(\omega)F_5(\omega)\Big)\\
   &=\Big(a_3(\omega)+E_7(\omega)\Big)+j\Big(-b_3(\omega)+F_7(\omega)\Big),
    \end{split}
\end{equation}
where $E_7(\omega)$ and $F_7(\omega)$ are $\mathcal{Z}_L$ functions. Therefore,
\begin{equation}
\begin{split}
    \mathcal{N}_y (\omega)=\mathfrak{R}\Big(M_1^{*}(j\omega)M_3(j\omega)\Big)=a_3(\omega)+E_7(\omega),
    \end{split}
\end{equation}
by replacing $a_3(\omega)$ from (\ref{eq: M2'}), we have
\begin{equation}
\begin{split}
    \mathcal{N}_y (\omega)=\frac{\omega_k \omega_r}{\omega^2+\omega_r ^{2}}+E_7(\omega).
    \end{split}
\end{equation}
Therefore, for the case $C_R=\text{CI}$ where $\omega_r= 0$, $\mathcal{N}_y (\omega)=E_7(\omega)$, which is a $\mathcal{Z}_L$ function. Thus, based on Lemma \ref{NxNyZl}, it is not possible to determine the stability of the reset control system \eqref{eq.SS closed loop} in the presence of delay by using Theorem \ref{lem:Hb frequency}.
However, in the case where $\omega_r\neq 0$ ($C_R=\text{GFORE}$), $\mathcal{N}_y$ is not a $\mathcal{Z}_L$ function and it is possible to apply Theorem \ref{lem:Hb frequency} to demonstrate the stability of a reset control system with GFORE element. To enhance comprehension, Fig. \ref{fig:NSV delay} depicts an example where both $\mathcal{N}_x$ and $\mathcal{N}_y$ are $\mathcal{Z}_L$ functions ($\omega_r= 0$), as well as a scenario where only $\mathcal{N}_x$ is a $\mathcal{Z}_L$ function ($\omega_r\neq 0$).
\end{proof}

    \begin{figure}[!t]
\centering
\subfloat[]{\includegraphics[width=0.5\columnwidth]{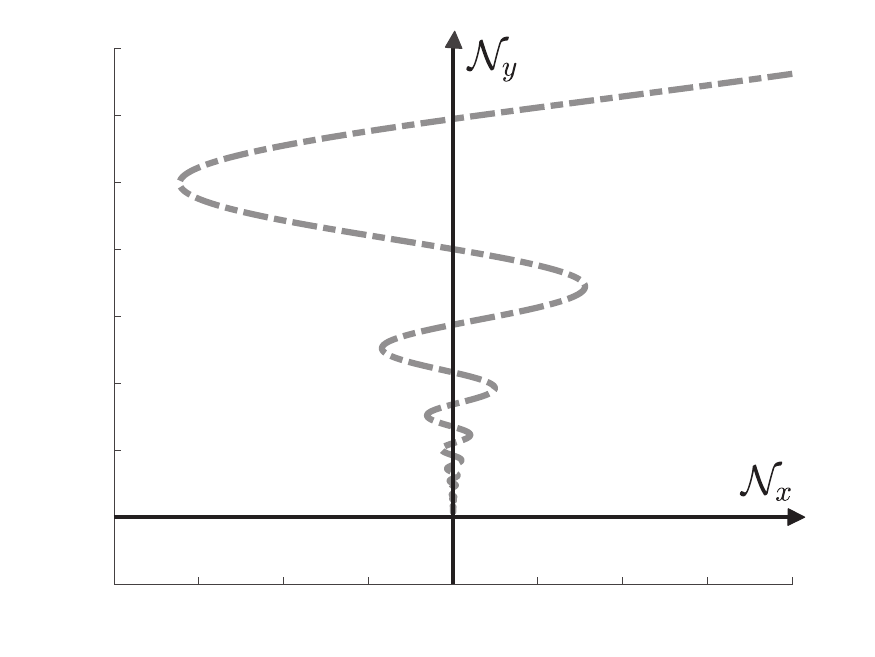}%
\label{fig:Nx=0}}
\subfloat[]{\includegraphics[width=0.5\columnwidth]{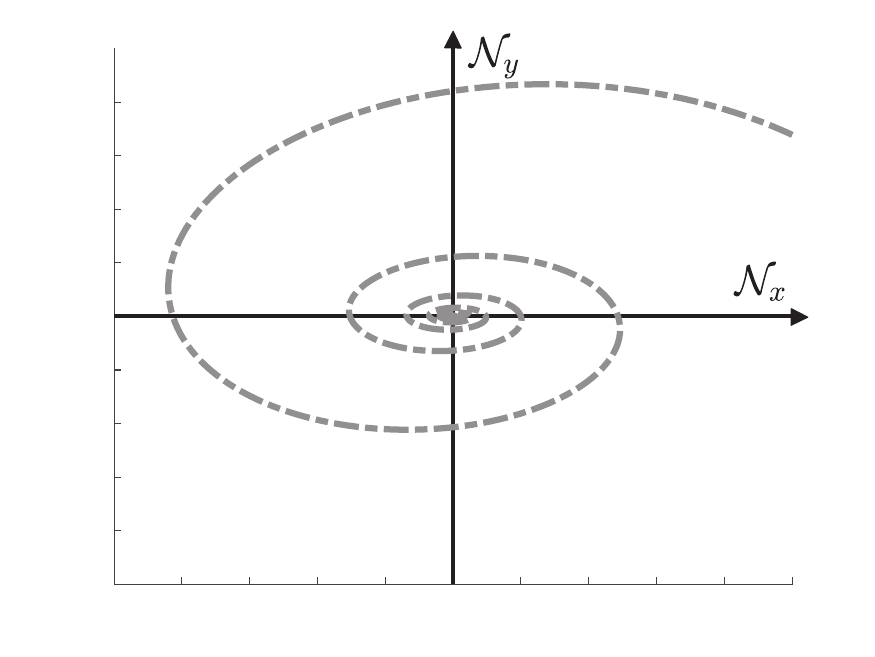}}%
\label{fig:nxny=0}
\caption{NSV plot ($\stackrel{\rightarrow}{\mathcal{N}}(\omega)$) in the presence of delay at high frequencies ($\omega \rightarrow \infty$), for (a) $\omega_r\neq 0$, and (b) $\omega_r$= 0.}
\label{fig:NSV delay}
\end{figure}
\section{Illustrative examples} \label{sec: example}
In this section, we test the validity of the proposed method for assessing the stability of a general RCS using the frequency response of its components. First, we evaluate the stability of a mass-spring-damper system, both with and without delay, to validate the results presented in Section \ref{sec: delay}. Additionally, we analyze the stability of an industrial precision positioning system, where only the FRF of the plant is available.

The closed-loop structure for both cases is the same as Fig. \ref{fig:Block diagram cl}, where
\begin{equation}
\begin{split}
    \label{eq:controllers}
    &C_R=\text{GFORE}\,(A_r=-\omega_r, B_r=1, C_r=\omega_k, D_r=0), \\
   &C_2(s)=k_p\;\omega_{i}\bigg(k_{g}+\frac{1}{s}\bigg)\bigg(\frac{\frac{s}{\omega_d}+1}{\frac{s}{\omega_t}+1}\bigg),\, C_1(s)=1,  \\
   &C_3(s)=\frac{s}{\Big(k_{g} s+1\Big)\omega_{i}},\quad C_s(s)=1,
\end{split}
\end{equation}
and
\begin{equation}
    k_{g}=\frac{1}{\omega_r|1+\frac{4j}{\pi}\frac{1-\gamma}{1+\gamma}|}.
\end{equation}
More details about the controller design and parameters tuning can be found in \cite{Yixuan}.

\subsection{Mass-spring-damper system}
This example demonstrates the effectiveness and repeatability of the stability method in a general form. The transfer function of a mass-spring-damper (MSD) system with transport delay is given by

\begin{equation}
\label{eq: M_S_D}
G(s) = \frac{\omega_n^2}{s^2 + 2\zeta\omega_n s + \omega_n^2} D(s, T)
\end{equation}
where $\zeta = 0.2$, $\omega_n = 30$, and $D(s, T)$ represents a time delay of $T$. Two scenarios are examined: one with a delay of $T = 0.0015\,$sec and one without any delay ($T = 0$). The controller structure remains consistent with the one presented in \eqref{eq:controllers}, and the corresponding parameters are listed in Table \ref{tab: parameters}.

In this example, the base linear system is stable, with $-1 < \gamma < 1$. Fig. \ref{fig: theta_N MSD} illustrates $\theta_{\mathcal{N}}(\omega)$ for both scenarios: with and without delay. It is clear that $-\frac{\pi}{2} < \theta_{\mathcal{N}}(\omega) < \pi$ and $(\theta_2 - \theta_1) < \pi$. Also, it can be seen that the delay only adds oscillation around $\theta_{\mathcal{N}}(\omega)=\frac{\pi}{2}$ ($\mathcal{N}_x (\omega)=0$) which means only $\mathcal{N}_x (\omega)$ becomes $\mathcal{Z}_L$ function and still we can assess the stability for this GFORE-based RCS like the result in Corollary \ref{CI,GFORE}. 

Consequently, based on Theorem \ref{lem:Hb frequency} and Lemma \ref{lemma UBIBS}, the closed-loop system in this case qualifies as a UBIBS stable reset control system, irrespective of the presence or absence of delay. Additionally, according to Lemma \ref{convrgence} and considering the zero initial condition for the designed reset controller, the RCS in this example has a uniformly exponentially convergent solution.

Note that here, for both cases (with and without delay), the same controller parameters are considered to be able to observe only the effect of delay. The difference in the mid-frequency range (the peak of $\theta_{\mathcal{N}}(\omega)$) could be related to the stability of the base linear system since the system with delay has less phase margin. However, despite observing this behavior (a more stable base linear system leads to a more robust $\theta_{\mathcal{N}}(\omega)$) in every example, a mathematical proof is still needed, which could be part of future studies.


\begin{table*}[]
\centering
\caption{Controller parameters.}
\label{tab: parameters}
\resizebox{\textwidth}{!}{%
\begin{tabular}{|c|c|c|c|c|c|c|c|c|}
\hline
          & $\gamma$ &$D_r$& $\omega_r$           &$\omega_k$& $\omega_i$ & $\omega_d$ & $\omega_t$ & $k_p$    \\ \hline
MSD       & 0        &0& 42.66                &42.66& 38.71      & 50         & 450       & 6.5     \\ \hline
Y$\Theta$ & 0        &0& 67.5$\times 10^{-4}$ &67.5$\times 10^{-4}$& 61.25$\times 10^{-4}$      & 79.167$\times 10^{-4}$     & 356.25$\times 10^{-4}$     & 3518300 \\ \hline
\end{tabular}
}
\end{table*}
\begin{figure}[!t]
\centering
\includegraphics[width=0.9\columnwidth]{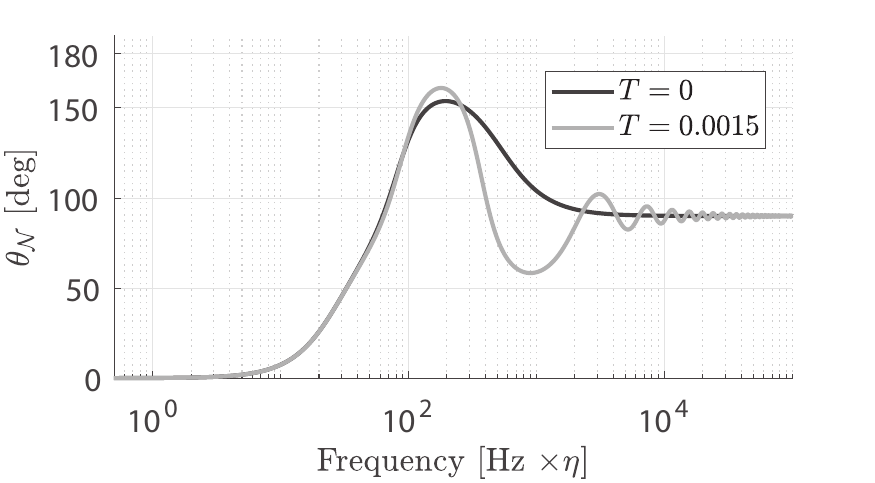}	\caption{$\theta_{\mathcal{N}}(\omega)$ for the MSD system.}
	\label{fig: theta_N MSD}
\end{figure}
\subsection{Precision positioning wire-bonding machine}
Here, the stability of a general reset control system is assessed, where the plant under control is the precision positioning stage of ASMPT's wire-bonding machine. It is assumed that the only information about the plant is its measured FRF, which is depicted in Fig. \ref{fig: FRF}. To maintain confidentiality, adjustments have been made to the frequency axis using an arbitrary constant $\eta$, while excluding both magnitude and phase information.

The control parameters for this case are presented in Table \ref{tab: parameters}. These parameters are obtained from \cite{Yixuan} based on an automated tuning algorithm that aims to maximize the open-loop bandwidth while considering the closed-loop performance.

\begin{figure}
	\centering
 	\includegraphics[width=0.9\columnwidth]{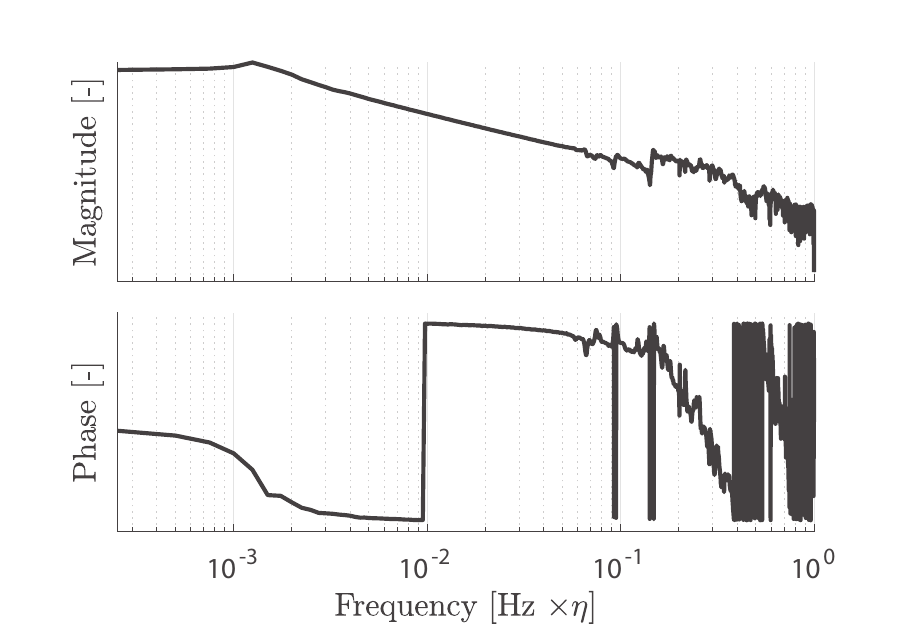}
	\caption{Frequency response data from the x-axis motion platform of the ASMPT's wire-bonding machine.}
	\label{fig: FRF}
\end{figure}

Considering the plant in Fig. \ref{fig: FRF}, controller elements in \eqref{eq:controllers}, and parameters in Table. \ref{tab: parameters}, the Theorem \ref{lem:Hb frequency} is applied to the closed-loop system. The controller is designed to have a stable base linear system and also $-1<\gamma<1$. The phase of the NSV ($\theta_{\mathcal{N}}$), shown in Fig. \ref{fig: theta_N}, reveals that $-\frac{\pi}{2} < \theta_{\mathcal{N}}(\omega) < \pi$ and $(\theta_2 - \theta_1) < \pi$. As a result, conditions three and four in Theorem \ref{lem:Hb frequency} for the GFORE case ($\omega_r \neq 0$) are satisfied. This indicates that the closed-loop system is globally uniformly asymptotically stable for the zero-input case.

Also, since the reset control system has been designed for the zero initial condition and $C_s=1$, by using Lemma \ref{lemma UBIBS} and Lemma \ref{convrgence}, the closed-loop system has the UBIBS property and a uniformly exponentially convergent solution for any input $w$ which is a Bohl function.

This example shows that even without a parametric model of the system or any transfer function, it is still possible to assess the stability of the most general form of a reset control system (including pre-, post-, and parallel filters along reset element in the loop) by using only the measured FRF of the plant.

\begin{figure}
	\centering
 	\includegraphics[width=0.9\columnwidth]{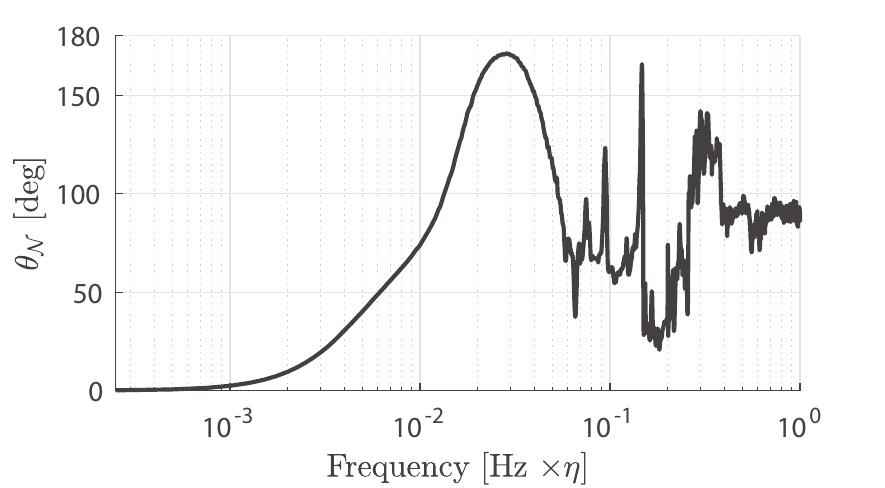}
	\caption{ $\theta_{\mathcal{N}}(\omega)$ for the precision positioning control system.}
	\label{fig: theta_N}
\end{figure}
\section{Conclusion} \label{sec: conclusion}
In this paper, we analytically prove the mapping between the matrix-based $H_\beta$ transfer function and an FRF-based one for general (parallel) RCSs. We then use this mapping to relax the LMI-based conditions and assess quadratic stability using only some graphical, frequency domain-based conditions.
Using the results of this study, we can now assess the stability of a parallel RCS or an RCS with a feed-through term, as well as the UBIBS and convergence properties of the system, achievements that were not possible with existing FRF-based methods. Additionally, we analyzed the effect of delay on this method. The results showed that in cases where the plant suffers from delay, only the FORE/GFORE element results in feasible decisions regarding stability. We validated this finding by considering a mass-spring-damper system in the presence and absence of delay for a controller using GFORE. Furthermore, we used another illustrative example to show the model-free characteristics of this method by considering an FRF of the industrial setup as a plant. This means that no matter which structure of reset control we use or which system we want to control (with or without delay), using this method, we are able to assess its stability. Further research will focus on examining the robustness of the current approach and defining specific stability margins, providing valuable insights into its performance under various operating conditions and uncertainties. Moreover, a comprehensive conservativeness analysis of the existing stability methods for RCSs is required.

\appendices
\section{Proof of Theorem \ref{Th. Theorem1}} \label{App: H_b}
\begin{proof}
 Using the approach from \cite[Proposition 4.5]{banos2012reset}, to demonstrate the quadratic stability of the system in \eqref{eq.SS closed loop}, we show the existence of a matrix $P > 0$ such that the quadratic Lyapunov function \( V(x(t)) = x^\top(t) P x(t) \) decreases over the entire state space along the system trajectories and is non-increasing at the reset jumps (see \cite[Theorem 1]{beker2004fundamental}). This leads to the following:
\begin{equation}
    \label{eq: Lyapunov A}
    \bar{A}^\top P+P\bar{A}<0, 
\end{equation}
and
\begin{equation}
    \label{eq: reset jumps lyapunov}
    x^\top (t)\left(A_\rho^\top P A_\rho -P \right)x(t)\leq 0, \: \forall x(t)\in \mathcal{F}, 
\end{equation}
where the reset surface $\mathcal{F}$ for the zero input case ($w(t)=0$) gives $\{x(t) \in \mathbb{R}^{n_l+1}:\bar{C}_ex(t)=C_e x_l(t)=0 \wedge (I-A_\rho)x(t)\neq0\}$ with $\bar{C}_e=\begin{bmatrix}
    0 & C_e
\end{bmatrix}$.\\
Considering $P=P^\top=\begin{bmatrix}
    P_1 & P_2 \\
    P_2^\top & P_3
\end{bmatrix}$, $A_\rho = \begin{bmatrix}
    \gamma & 0_{1\times n_l}\\
    
    0_{n_l\times 1} & I_{n_l\times n_l}
\end{bmatrix}$, and $x(t)=\begin{bmatrix}
    x_r(t)^T & x_l(t)^T
\end{bmatrix}^T$, for \eqref{eq: reset jumps lyapunov} we have:
\begin{equation}
\begin{split}
    \label{eq: extended lyapunov}
     &x^\top (t)\left(A_\rho^\top P A_\rho -P \right)x(t)\\
     &=(\gamma^2-1)x_r^2(t)P_1+(\gamma-1)x_r(t)x_l^\top(t)P_2^\top\\
     &+(\gamma-1)x_r(t)P_2x_l(t).
     \end{split}
\end{equation}
By selecting $P_2=\beta C_e$ where $\beta \in \mathbb{R}$, we have:
\begin{equation}
\begin{split}
    \label{eq: extended lyapunov 2}
     &x^\top (t)\left(A_\rho^\top P A_\rho -P \right)x(t)\\
     &=(\gamma^2-1)x_r^2(t)P_1+(\gamma-1)x_r(t)\beta x_l^\top(t) C_e^\top\\
     &+(\gamma-1)x_r(t)\beta C_ex_l(t).
     \end{split}
\end{equation}
Given that $C_ex_l(t)=0$ at the reset jumps, the expression simplifies to:
\begin{equation}
    \label{eq: extended lyapunov 3}
     x^\top (t)\left(A_\rho^\top P A_\rho -P \right)x(t)\\
     =(\gamma^2-1)x_r^2(t)P_1,
\end{equation}
where for $P_1>0$ and $(\gamma^2-1)\leq 0$ (or $-1\leq \gamma \leq 1$) it results
\begin{equation}
    \label{eq: reset jumps lyapunov 2}
    x^\top (t)\left(A_\rho^\top P A_\rho -P \right)x(t)\leq 0, \: \forall x(t)\in \mathcal{F}.
\end{equation}
This implies that the inequality in \eqref{eq: reset jumps lyapunov} is satisfied for every \(-1 \leq \gamma \leq 1\) and for matrices \(P > 0\) of the following form:
\begin{equation}
    \label{eq: P matrix}
    B_0^\top P=C_0,
\end{equation}
where
\begin{equation}
B_0= \begin{bmatrix}
    1 \\
    0_{n_l \times 1}
\end{bmatrix}, \quad
 C_0=\begin{bmatrix}
\varrho & \beta C_e
\end{bmatrix}.
\label{eq. C0B0 proof}
\end{equation}
with $\varrho=P_1 \in \mathbb{R}^{+}$. Thus, for \eqref{eq: Lyapunov A} and \eqref{eq: P matrix} we can write:
\begin{equation}
    \label{eq: KYP}
\begin{bmatrix}
    \bar{A}^\top P+P\bar{A}+2\varepsilon P & PB_0-C_0^\top \\
     B_0^\top P-C_0 & 0
\end{bmatrix}\leq0,
\end{equation}
with $\exists\, P>0$ and $\exists \,\varepsilon>0$. The remainder of the proof follows that of \cite[Proposition 4.5]{banos2012reset} (from equation (4.30) onward), utilizing the generalized Kalman–Yakubovich–Popov (KYP) lemma \cite{rantzer1996kalman}.

\end{proof}

\section{Proof of Lemma \ref{lem: Hb transfer}}\label{App: I}
\begin{proof}
In this proof, the goal is to show that the transfer function of $H_\beta(s)$ in (\ref{eq.H beta}) is equal to the transfer function in (\ref{eq: H_b FRF}). By starting from (\ref{eq.H beta}) for $(sI-\bar{A})$ we have
\begin{equation}
\label{eq: sI-A}
(sI - \bar{A}) = \begin{bmatrix}
s-A_r & -B_rC_u \\
-B_uC_r & sI-A-B_uD_rC_u \\
\end{bmatrix}=\begin{bmatrix}
Q_1 & Q_2 \\
Q_3 & Q_4 \\
\end{bmatrix}.
\end{equation}
Defining
\begin{equation}
\label{eq: M-1}
(sI - \bar{A})^{-1} = \begin{bmatrix}
W & X \\
Y & Z
\end{bmatrix},
\end{equation}
and substituting (\ref{eq: M-1}) in (\ref{eq.H beta}), gives
\begin{equation}
H_\beta(s) = \begin{bmatrix}
\varrho & \beta C_e
\end{bmatrix}
\begin{bmatrix}
W & X \\
Y & Z
\end{bmatrix}
\begin{bmatrix}
1 \\
0
\end{bmatrix}
= \varrho W + \beta C_e Y.
\label{eq: Hb WY}
\end{equation}
Therefore, only $W$ and $Y$ are needed to be calculated.
By using Lemma \ref{lem: Mat inv} we have
\begin{equation}
W = \left(Q_1 - Q_2Q_4^{-1}Q_3\right)^{-1},
\label{eq: W}
\end{equation}
and
\begin{equation}
Y = -Q_4^{-1}Q_3W.
\label{eq: WY}
\end{equation}
For $W$, using \eqref{eq: sI-A} gives
\begin{equation}
\begin{split}
W &= \left(Q_1 - Q_2 Q_4^{-1} Q_3\right)^{-1} \\ &= \left(s - A_r - B_r C_u (sI - A - B_u D_r C_u)^{-1} B_u C_r\right)^{-1}.
\end{split}
\label{eq: W expand}
\end{equation}
Now considering $K=Q_1$, $U=Q_2$, $J=-Q_4^{-1}$, and $V=Q_3$, by applying Lemma \ref{lem: Woodbury} we have
 \begin{multline}
W = \left(s - A_r\right)^{-1} - \left(s - A_r\right)^{-1} B_r C_u \\ ...\left(-(sI - A) + B_u D_r C_u + B_u C_r \left(s - A_r\right)^{-1}B_r C_u\right)^{-1}\\...B_u C_r \left(s - A_r\right)^{-1}.
\label{eq: W expand 2}
 \end{multline}
Having $R(s)=C_r(sI-A_r)^{-1}B_r+D_r$, the part
\begin{equation}
\left(-(sI - A) + B_u D_r C_u + B_u C_r \left(s - A_r\right)^{-1}B_r C_u\right)^{-1},
\end{equation}
can be rewritten as
\begin{equation}
\label{dfr}
\begin{split}
&\Big(-(sI - A) + B_u D_r C_u + B_u (R(s)-D_r) C_u\Big)^{-1} \\ &= \Big(-(sI - A) + B_u \left(D_r + R(s) - D_r\right) C_u\Big)^{-1}  \\
&= \left(-(sI - A) + B_u R(s) C_u\right)^{-1}.
\end{split}
\end{equation}
Lemma \ref{lem: Woodbury} can also be applied to \eqref{dfr} as follows
\begin{equation}
\begin{split}
&\left(-(sI - A) + B_u R(s) C_u\right)^{-1} \\
&=-\left(sI - A\right)^{-1} - \left(sI - A\right)^{-1} B_u\\... &\left(\frac{1}{{R(s)}} - C_u \left(sI - A\right)^{-1} B_u\right)^{-1} C_u \left(sI - A\right)^{-1}.
\end{split}
\label{eq: w expand}
\end{equation}
For the part $C_u (sI - A)^{-1} B_u$, from system description in (\ref{eq.SS linear}) when $r=d=0$, we have
\begin{equation}
\begin{split}
C_u (sI - A)^{-1} B_u = \frac{U_1(s)}{U_r(s)} = \frac{-C_1(s) C_2(s) G(s)}{1 + C_1(s) C_2(s) G(s) C_3(s)},
\label{eq: P}
\end{split}
\end{equation}
where capitalized variables $U_1$ and $U_r$ are the Laplace transforms of the respective non-capitalized time-domain signals. For simplicity, $\frac{U_1(s)}{U_r(s)}$ is considered as $P(s)$. Therefore, (\ref{eq: w expand}) can be rewritten as follows
\begin{equation}
\begin{split}
&\left(-(sI - A) + B_u R(s) C_u\right)^{-1}=-(sI - A)^{-1}-\\
&(sI - A)^{-1} B_u \left(\frac{R(s)}{1 - P(s)R(s)}\right)C_u (sI - A)^{-1}.
\label{eq: w expand 2}
\end{split}
\end{equation}
Now by substituting (\ref{eq: w expand 2}) in (\ref{eq: W expand 2}), and considering $(s - A_r)^{-1} = \frac{R(s) - D_r}{C_r B_r}$, it gives
\begin{equation}
\begin{split}
W =& \frac{R(s) - D_r}{C_r B_r} - \frac{R(s) - D_r}{C_r B_r} B_r  \Bigg(-C_u(sI - A)^{-1}B_u\\ &- C_u(sI - A)^{-1} B_u \left(\frac{R(s)}{1 - P(s) R(s)}\right)\\...&C_u(sI - A)^{-1}B_u\Bigg)C_r \frac{R(s) - D_r}{C_r B_r},
\end{split}
\end{equation}
where again by considering $C_u (sI - A)^{-1} B_u=P(s)$, it leads to
\begin{equation}
\begin{split}
W =& \frac{R(s) - D_r}{C_r B_r} - \frac{R(s) - D_r}{C_r} \Bigg(-P(s)\\ &-P(s) \Big(\frac{R(s)}{1 - P(s) R(s)}\Big) P(s)\Bigg) \frac{R(s) - D_r}{B_r} \\
=& \left(\frac{R(s) - D_r}{C_r B_r}\right) \left(\frac{1 - P(s) D_r}{1 - P(s) R(s)}\right).
\end{split}
\label{eq: W main}
\end{equation}
Now $W$ is calculated based on known parameters and transfer functions.

For $Y$ from \eqref{eq: WY} we have 
\begin{equation}
\begin{split}
Y = -Q_4^{-1} Q_3 W = (sI - A - B_u D_r C_u)^{-1} B_u C_r W,
\end{split}
\end{equation}
by applying Lemma \ref{lem: Woodbury}, it gives
\begin{equation}
\begin{split}
Y =& (sI - A)^{-1} B_u C_r W - (sI - A)^{-1} B_u \left(\frac{D_r}{P(s) D_r - 1}\right)\\... &C_u (sI - A)^{-1} B_u C_r W.
\end{split}
\end{equation}
Thus, based on \eqref{eq: Hb WY}, for $\beta C_u Y$ we have
\begin{equation}
\begin{split}
\beta C_e Y =& \beta C_e (sI - A)^{-1} B_u C_r W - \beta C_e (sI - A)^{-1} B_u\\... &\left(\frac{D_r}{P(s) D_r - 1}\right) C_u (sI - A)^{-1} B_u C_r W,
\label{eq: Y expand}
\end{split}
\end{equation}
where again from the state-space description of the LTI part of the system in (\ref{eq.SS linear}), for $C_e (sI - A)^{-1} B_u$ we have
\begin{equation}
\begin{split}
&C_e (sI - A)^{-1} B_u = \frac{E_r(s)}{U_r(s)} \\ &= \frac{-C_1(s) C_2(s) G(s)C_s(s)}{1 + C_1(s) C_2(s) G(s) C_3(s)}=P(s)C_s(s).
\label{eq: PCs}
\end{split}
\end{equation}
Thus, by replacing (\ref{eq: P}) in (\ref{eq: Y expand}), it gives
\begin{equation}
\beta C_e Y = \beta C_r W \left(\frac{-P(s)C_s(s)}{P(s) D_r - 1}\right).\\
\end{equation}
Therefore, $H_\beta(s)$ in \eqref{eq: Hb WY} can be written as
\begin{equation}
\begin{split}
H_\beta(s)=&\varrho W + \beta C_e Y=\beta^{'} (R(s) - D_r)\left(\frac{P(s)C_s(s)}{1 - P(s) R(s)}\right) \\ &+\varrho^{'} (R(s) - D_r)\left(\frac{1 - P(s) D_r}{1 - P(s) R(s)}\right),
\label{eq: Hb semifinal}
\end{split}
\end{equation}
where $\beta^{'}=-\frac{\beta}{B_r}$, and $\varrho^{'}=\frac{\varrho}{B_rC_r}$ (we assume $B_rC_r>0$, which is a relavant assumption for GFORE, CI and PCI elements). By replacing $P(s)=\frac{-L(s)}{1 + L(s) C_3(s)}$ and $L(s)= C_1(s) C_2(s) G(s)$ , it yields
\begin{equation}
    \label{eq: H_b main}
    \resizebox{1\hsize}{!}{
    $H_\beta(s)=\frac{\beta^{'}L(s)C_s(s)\Big(R(s)-D_r\Big)+\varrho^{'}\bigg(1+L(s)\Big(C_3(s)+D_r\Big)\bigg)\Bigl(R(s)-D_r\Bigl)}{1+L(s)\Bigl(R(s)+C_3(s)\Bigl)},$
    }
    \end{equation}
which is equal to the transfer function in \eqref{eq: H_b FRF}.
\end{proof}

\section{proof of Theorem \ref{lem:Hb frequency}}\label{App: II}
\begin{proof}
    According to Theorem \ref{Th. Theorem1} to ensure that the reset control system (\ref{eq.SS closed loop}) is globally uniformly asymptotically stable, it must be shown that the $H_\beta(s)$ is SPR, $(\bar{A},B_0)$ is controllable, $(\bar{A},C_0)$ is observable, and $-1<\gamma<1$. The required conditions for a transfer function ($p\times p$) to be SPR are presented in Lemma \ref{lem: SPR}. Regarding the first condition in Lemma \ref{lem: SPR}, the transfer function \( H_\beta \) must satisfy the Horowitz criterion. Since \( H_\beta \) and the base linear transfer function share the same denominator, it follows from the expression for \( H_\beta \) in \eqref{eq: H_b FRF} that if both the base linear system and the shaping filter \( C_s(s) \) are stable, then the first condition in Lemma \ref{lem: SPR} is satisfied.
    
    Furthermore, because the $H_\beta$ transfer function is a single-input and single-output transfer function, the second and third conditions in Lemma \ref{lem: SPR} can be expressed as steps 1 and 2, respectively. Also, controllability and observability of $(\bar{A},B_0)$ and $(\bar{A},C_0)$ will investigate in step 3.
\begin{itemize}
    \item Step 1: It is shown that there is a $\beta$ and $\varrho>0$ such that $\mathfrak{R}(H_\beta(j\omega)) > 0$ for all $\omega \in \mathbb{R}$.
    \item Step 2: It is shown that either $\lim_{s \to \infty}H_\beta (s) >0$ or $\lim_{s \to \infty}H_\beta (s) =0$ and $\lim_{\omega \to \infty} \omega^2\mathfrak{R}(H_\beta(j\omega))>0$.
    \item Step 3: It is shown that $(\bar{A},B_0)$ and $(\bar{A},C_0)$ are controllable and observable respectively.\\
    \label{th. theorem2}
\end{itemize}

\textbf{Step 1:}
To do this, first, it is necessary to calculate the real part of $H_\beta(j\omega)$ in (\ref{eq: H_b FRF}). By utilizing the notation in (\ref{eq.M}) for $H_\beta(j\omega)$, we can then proceed as follows
\begin{equation}
    \label{eq: nsv p 1}
    H_\beta(j\omega)=\frac{\beta^{'}M_2+\varrho^{'}M_3}{M_1},
\end{equation}
multiplying both the numerator and the denominator by the complex conjugate of the denominator ($M_1^{*}$), yields
\begin{equation}
    \label{eq: nsv p 2}
    H_\beta(j\omega)=\frac{\beta^{'}M_2M_1^{*}+\varrho^{'}M_3M_1^{*}}{M_1M_1^{*}},
\end{equation}
where $\mathfrak{R}(M_1M_1^{*})>0$, and $I(M_1M_1^{*})=0$ ($I(.)$ means the imaginary part). Thus 
\begin{equation}
    \label{eq: nsv p 25}
\mathfrak{R}\Big(H_\beta(j\omega)\Big)=\frac{\beta^{'}\mathfrak{R}\Big(M_2M_1^{*}\Big)+\varrho^{'}\mathfrak{R}\Big(M_3M_1^{*}\Big)}{M_1M_1^{*}}.
\end{equation}
To have $\mathfrak{R}\Big(H_\beta(j\omega)\Big)>0$ it is needed to show that 
\begin{equation}
    \label{eq: nsv p 3}
\beta^{'}\mathfrak{R}\Big(M_2M_1^{*}\Big)+\varrho^{'}\mathfrak{R}\Big(M_3M_1^{*}\Big)>0.
\end{equation}
Considering $\overrightarrow{\mathcal{N}}(\omega)$ as (\ref{eq.NSV}), and defining $\overrightarrow{\xi}=\begin{bmatrix}
    \beta^{'} & \varrho^{'}
\end{bmatrix}$, the equation (\ref{eq: nsv p 3}) can be rewritten as
\begin{equation}
    \label{eq: N xi}
    \overrightarrow{\xi}.\overrightarrow{\mathcal{N}}(\omega)>0, \quad \forall \, \omega \in [0,\infty).
\end{equation}
Having $\theta_{\xi}=\phase{\overrightarrow{\xi}}$, and $\theta_{\mathcal{N}}(\omega)=\phase{\overrightarrow{\mathcal{N}}(\omega)}$, we can write \eqref{eq: N xi} as
\begin{equation}
    \label{eq: N xi 2}
    |\overrightarrow{\xi}||\overrightarrow{\mathcal{N}}(\omega)| \cos{\Big(\theta_{\xi}-\theta_{\mathcal{N}}(\omega)\Big)}>0, \;|\stackrel{\rightarrow}{\xi}|\neq0, |\stackrel{\rightarrow}{\mathcal{N}}|\neq0.\\
\end{equation}
Thus, to have $\overrightarrow{\xi}.\overrightarrow{\mathcal{N}}(\omega)>0$, we should have
\begin{equation}
    \label{eq: N xi 3}
    \cos{\Big(\theta_{\xi}-\theta_{\mathcal{N}}(\omega)\Big)}>0, \quad \text{for all}\quad \omega \in [0,\infty).
\end{equation}
Since $\cos{(x)}>0$ yields $-\frac{\pi}{2}<x<\frac{\pi}{2}$, we should have
\begin{equation}
   -\frac{\pi}{2}<\theta_{\xi}-\theta_{\mathcal{N}}(\omega)<\frac{\pi}{2}, \quad \forall \omega \in [0,\infty).
   \label{eq: zettaN 2}
\end{equation}
Thus, knowing $-\infty<\beta^{'}<\infty$ and $\varrho^{'}>0$, it gives
\begin{equation}
    0<\theta_{\xi}<\pi.
    \label{eq: zetta}
\end{equation}
Let $\theta_1 = \smash{\displaystyle\min_{\forall \omega \in \mathbb{R}}} \theta_{\mathcal{N}}(\omega)$ and $\theta_2 = \smash{\displaystyle\max_{\forall \omega \in \mathbb{R}}} \theta_{\mathcal{N}}(\omega)$ (see Definition \ref{def:Types}). In reference to \eqref{eq: zetta}, it follows that $(\theta_2 - \theta_1) < \pi$ must hold to satisfy \eqref{eq: zettaN 2}.

Furthermore, if $\theta_{\mathcal{N}}(\omega)$ lies in both intervals $[\pi, \frac{3\pi}{2})$ and $[-\frac{\pi}{2}, 0)$, it is evident that no $0 < \theta_{\xi} < \pi$ can satisfy \eqref{eq: zettaN 2}.
\\

Hence, the conditions for $\overrightarrow{\mathcal{N}}(\omega)$ to satisfy (\ref{eq: N xi}) are as follows
\begin{itemize}
    \item $-\frac{\pi}{2}<\theta_{\mathcal{N}}(\omega)<\pi$ and/or $0<\theta_{\mathcal{N}}(\omega)<\frac{3\pi}{2}$, $\quad \forall \, \omega \in [0,\infty)$.
    \item $(\theta_2-\theta_1)<\pi$. \\
\end{itemize}

\textbf{Step 2:} Regarding the $H_\beta$ transfer function in \eqref{eq: H_b FRF}, we have
\begin{equation}
\begin{split}
    \label{ Hb infty}
    \lim_{s \to \infty} H_\beta(s)=\beta^{'}L(s)C_s(s)\big(R(s)-D_r\big)\\+\varrho^{'}\big(R(s)-D_r\big),
    \end{split}
\end{equation}
where either $\lim_{s \to \infty}H_\beta (s) >0$ or $\lim_{s \to \infty}H_\beta (s) =0$ and $\lim_{\omega \to \infty}$ need to be satisfied for both cases with $\omega_r\neq 0$ and $\omega_r=0$.\\

\begin{itemize}
    \item $\omega_r\neq0$ ($R(s)=\frac{\omega_k}{s+\omega_r}+D_r$) \\
    \begin{itemize}
        \item[*] $n - m = 1$ (the relative degree of ${L}(s)C_s(s)$)\\
        \begin{equation}
    \lim_{\omega \to \infty} \omega^2\mathfrak{R}(H_\beta(j\omega))=-\beta^{'}K+\varrho^{'}\omega_r \omega_k=\overrightarrow{\xi} . \overrightarrow{\mathcal{N'}},
\end{equation}
by setting $\overrightarrow{\mathcal{N'}}=\begin{bmatrix}
     -K & \omega_r \omega_k
    \end{bmatrix}^{T}$, we have
    \begin{equation}
        \phase{\overrightarrow{\mathcal{N'}}}=\lim_{\omega \to \infty}\phase{\overrightarrow{\mathcal{N}}(\omega)}\xRightarrow{(\ref{eq: teta})}\theta_1 \leq \phase{\stackrel{\rightarrow}{\mathcal{N'}}} \leq \theta_2,
    \label{eq: N_theta}
    \end{equation}
where, by using step 1, and starting from (\ref{eq: N xi}) for $\overrightarrow{\xi} .  \overrightarrow{\mathcal{N'}}$, it gives
\begin{equation}
    \lim_{\omega \to \infty} \omega^2\mathfrak{R}(H_\beta(j\omega))=\overrightarrow{\xi} . \overrightarrow{\mathcal{N'}}>0.
\end{equation}
\item[*] $n-m> 1$ \\
\begin{equation}
    \lim_{\omega \to \infty} \omega^2\mathfrak{R}(H_\beta(j\omega))=\varrho^{'}\omega_r \omega_k>0.
    \end{equation}
    \end{itemize}
\item $\omega_r=0$ \\
    \begin{itemize}
        \item[*] $n-m>1$ \\
        \begin{equation}
            \lim_{\omega \to \infty} \omega^2\mathfrak{R}(H_\beta(j\omega))=0,
        \end{equation}
        This implies that $H(j\omega)$ is not SPR when $n - m > 1$.
        \item[*] $n - m = 1$ \\
        \begin{equation}
    \lim_{\omega \to \infty} \omega^2\mathfrak{R}(H_\beta(j\omega))=-\beta^{'}\frac{K_{n}}{K_{m}}>0,
    \end{equation}
    \end{itemize}
\end{itemize}
means that in this case, the relative degree can only be 1, and $-\beta^{'}\frac{K_{n}}{K_{m}}>0$. Regarding the transfer function in \eqref{eq: L(s)Cs}, for $n-m=1$ we have $L(\infty)C_s(\infty)=\frac{K_{n}}{K_{m}s}$. Which leads to the following conditions:

    \begin{itemize}
    \item The relative degree of the transfer function $L(s)  C_s(s) = C_1(s)  C_2(s)  G(s) C_s(s)$ must be 1.
    \item If $\lim_{s \to \infty} \operatorname{phase}\left(L(s) C_s(s)\right) = -90$ ($\frac{K_{n}}{K_{m}} > 0$), then $0 < \theta_{\mathcal{N}}(\omega) < \frac{3\pi}{2}$.
    \item If $\lim_{s \to \infty} \operatorname{phase}\left(L(s) C_s(s)\right) = -270$ ($\frac{K_{n}}{K_{m}} < 0$), then $-\frac{\pi}{2} < \theta_{\mathcal{N}}(\omega) < \pi$.
\end{itemize}

\textbf{Step 3:}
In order to show that the pairs $(\bar{A},C_0)$ and $(\bar{A},B_0)$ are observable and controllable, respectively, it is sufficient to show that the denominator and the numerator of $H_\beta (j\omega)$ do not have any common root. Let $a_0 + jb_0$ be a root of the denominator. Then considering the $H_\beta(j\omega)=\frac{\beta^{'}M_2+\varrho^{'}M_3}{M_1}$ in \eqref{eq: nsv p 1}, we should have
\begin{equation}
    M_1(a_0,b_0)=0,
\end{equation}
where for the controllability and observability, the numerator must not have any root at $a_0 + jb_0$, which means
\begin{equation}
   \beta^{'} M_2(a_0,b_0)+\varrho^{'}M_3(a_0,b_0)\neq0.
\end{equation}
Here we assume there is a pair of ($\beta_1 ^{'}, \varrho_1 ^{'}$) such that
 \begin{equation}
 \label{eq: old pair}
    \beta_1 ' M_2(a_0,b_0)+\varrho_1 'M_3(a_0,b_0)=0.
\end{equation}
First, we consider the case that $M_2(a_0,b_0)\neq0$, and/or $M_3(a_0,b_0)\neq0$. From step 1 it can be shown that the eligible pairs for ($\beta^{'}, \varrho^{'}$) always have the following property
\begin{equation}
   \phase{(\beta ^{'}, \varrho ^{'})} \in \begin{bmatrix}
    \theta_2-\frac{\pi}{2}, \theta_1+\frac{\pi}{2}
\end{bmatrix},
\end{equation}
thus,
\begin{equation}
   \phase{(\beta_1 ^{'}, \varrho_1 ^{'})} \in \begin{bmatrix}
    \theta_2-\frac{\pi}{2}, \theta_1+\frac{\pi}{2}
\end{bmatrix}.
\end{equation}
Regarding the condition $(\theta_2-\theta_1)<\pi$ in step 1, it conclude that
\begin{equation}
    \begin{bmatrix}
    \theta_2-\frac{\pi}{2}, \theta_1+\frac{\pi}{2}
\end{bmatrix}\neq \varnothing.
\end{equation}
Therefore, a new pair $(\beta_2^{'}, \varrho_1^{'})=(\beta_1^{'}+\varepsilon, \varrho_1^{'})$ for $\varepsilon>0$ can be found that
\begin{equation}
   \phase{(\beta_2 ^{'}, \varrho_1 ^{'})} \in \begin{bmatrix}
    \theta_2-\frac{\pi}{2}, \theta_1+\frac{\pi}{2}
\end{bmatrix}.
\end{equation}
Substituting the new pairs in (\ref{eq: old pair}) yields,
\begin{align}
\label{eq: neq}
    &\beta_2 ' M_2(a_0,b_0)+\varrho_1 'M_3(a_0,b_0) \\ \nonumber
    &=(\beta_1^{'}+\varepsilon)M_2(a_0,b_0)+\varrho_1 'M_3(a_0,b_0) \\ \nonumber
     &=\beta_1^{'}M_2(a_0,b_0)+\varrho_1 'M_3(a_0,b_0)+\varepsilon M_2(a_0,b_0) \\ \nonumber
      &=\varepsilon M_2(a_0,b_0)\neq0.
\end{align}
Thus, for the case that $M_2(a_0,b_0)\neq0$, and/or $M_3(a_0,b_0)\neq0$, it is possible to find a pair ($\beta^{'}, \varrho^{'}$) such that $H(j\omega)$ be SPR and does not have any pole-zero cancellation.\\

Now considering the case that $M_2(a_0,b_0)=M_3(a_0,b_0)=M_1(a_0,b_0)=0$. Therefore, if we show that when $M_1(a_0,b_0)=0$, always one of the transfer functions $M_2(a_0,b_0)$ or $M_3(a_0,b_0)$ is non zero, then there is not any pole-zero cancellation and the proof is done. Thus, consider
\begin{equation}
M_1(a_0,b_0)=1+L(a_0,b_0)\Big(C_3(a_0,b_0)+R(a_0,b_0)\Big)=0
\label{eq: M1=0}
\end{equation}
\begin{equation}
\begin{split}
M_3(a_0,b_0)=&\bigg(1+L(a_0,b_0)\Big(C_3(a_0,b_0)+D_r\Big)\bigg)\\...&\Bigl(R(a_0,b_0)-D_r\Bigl)=0,
\label{eq: M3=0}
    \end{split}
\end{equation}
where (\ref{eq: M3=0}) yields three cases,
\begin{equation}
    \label{eq: case I}
   \text{I}:
    \begin{cases}
        1+L(a_0,b_0)\Big(C_3(a_0,b_0)+D_r\Big)=0,\\
        R(a_0,b_0)-D_r\neq0,
    \end{cases}
\end{equation}
\begin{equation}
    \label{eq: case II}
    \text{II}:
    \begin{cases}
        1+L(a_0,b_0)\Big(C_3(a_0,b_0)+D_r\Big)\neq0,\\
        R(a_0,b_0)-D_r=0,
    \end{cases}
\end{equation}
\begin{equation}
    \label{eq: case III}
    \text{III}:
    \begin{cases}
        1+L(a_0,b_0)\Big(C_3(a_0,b_0)+D_r\Big)=0,\\
        R(a_0,b_0)-D_r=0.
    \end{cases}
\end{equation}
For the case I, we have $ 1+L(a_0,b_0)\Big(C_3(a_0,b_0)+D_r\Big)=0$, where regarding $M_1$ in \eqref{eq: M1=0}, it yields
\begin{equation}
    R(a_0,b_0)=D_r
\end{equation}
which is not possible to have $R(a_0, b_0) - D_r = 0$ in case I. For cases II and III, $R(a_0, b_0) - D_r = 0$, implying
\begin{equation}
    \frac{\omega_k}{s+\omega_r}+D_r=D_r
\end{equation}
where $ \frac{\omega_k}{s+\omega_r}$ is a strictly proper first-order transfer function and can not be zero. Thus, it is also not possible to have $M_3(a_0, b_0) = 0$ and $M_1(a_0, b_0) = 0$ in these cases. Consequently, it is concluded that there is no pole-zero cancellation, and the pairs $(\bar{A}, C_0)$ and $(\bar{A}, B_0)$ are observable and controllable, respectively. \\

Therefore, regarding the Theorem \ref{Th. Theorem1} and steps 1 and 2, which involve the assessment of the strict positive realness of $H_\beta (s)$, and step 3, which pertains to the controllability and observability of $(\bar{A},B_0)$ and $(\bar{A},C_0)$, the $H_\beta$ condition is satisfied for the reset control system in (\ref{eq.SS closed loop}) by setting $-1<\gamma<1$. Consequently, the zero equilibrium of the reset control system in (\ref{eq.SS closed loop}) achieves global uniform asymptotic stability when $w = 0$. 
\end{proof}
\section{proof of Lemma \ref{lemma: sum_mult Zl}}\label{App: III}
\begin{proof}
    First we show that if $K_1(x_1)$ and $K_2(x_2)$ are $\mathcal{Z}_L$ functions then $K_1(x_1)+ K_2(x_2)$ is also a $\mathcal{Z}_L$ function. Consider $K_1(x_1)=g_1(x_1)h_1(x_1)$ and $K_2(x_2)=g_2(x_2)h_2(x_2)$ where
    \begin{equation}
    h_1(x_1)=\sin{(x_1)},
    \end{equation}
    and
    \begin{equation}
    h_2(x_2)=\sin{(x_2)}.
    \end{equation}
Thus,
\begin{equation}
    K_1(x_1)+ K_2(x_2)=g_1(x_1)\sin{(x_1)}+g_2(x_2)\sin{(x_2)}.
\end{equation}
Using the identity in \cite{SINaSINb}, we have
\begin{equation}
\label{eq k1+k2}
    K_1(x_1)+ K_2(x_2)=\sqrt{Q^2+R^2}\sin{(P+\Phi)},
\end{equation}
where
\begin{equation}
\label{eq: QRP}
    \begin{split}
    Q=&\left(g_1(x_1)+g_2(x_2)\right)\sin{\left(\frac{x_1+x_2}{2}\right)},\\
        R=&\left(g_1(x_1)-g_2(x_2)\right)\cos{\left(\frac{x_1+x_2}{2}\right)},\\
        P=&\frac{x_1-x_2}{2},\\
        \Phi=&\tan^{-1}{\left(\frac{Q}{R}\right)}.
    \end{split}
    \end{equation}
Then from \eqref{eq: QRP}, it can be shown $\lim_{x_1,x_2 \to \infty} \sqrt{Q^2+R^2}=0$. Thus the function $ K_1(x_1)+ K_2(x_2)$ in \eqref{eq k1+k2} is a $\mathcal{Z}_L$ function.\\
For $ K_1(x_1)\times K_2(x_2)=g_1(x_1)h_1(x_1)g_2(x_2)h_2(x_2)$ we can write
\begin{equation}
    K_1(x_1)\times K_2(x_2)=g_1(x_1)g_2(x_2)\sin{(x_1)}\sin{(x_2)},
\end{equation}
where 
\begin{equation}
    \sin{(x_1)}\sin{(x_2)}=\frac{1}{2}\left(\cos{(x_1-x_2)}-\cos{(x_1+x_2)}\right).
\end{equation}
Since every $\cos(y)$ function can be written in $\sin(y+\Phi_y)$ form by a shift in phase, the rest of the proof for $K_1(x_1)\times K_2(x_2)$ is the same as $K_1(x_1)+ K_2(x_2)$. Then, we conclude that if $K_1(x_1)$ and $K_2(x_2)$ are $\mathcal{Z}_L$ functions then $K_1(x_1)+ K_2(x_2)$ and $K_1(x_1)\times K_2(x_2)$ are also $\mathcal{Z}_L$ functions.
\end{proof}

\section*{Acknowledgment}

The authors sincerely appreciate the invaluable collaboration, insightful contributions, and generous support of Luke F. van Eijk and Dragan Kosti\'c from ASMPT throughout this project.
\section*{References}
\bibliographystyle{unsrt}
\bibliography{References.bib}    
\begin{IEEEbiography}[{\includegraphics[width=1in,height=1.25in,clip,keepaspectratio]{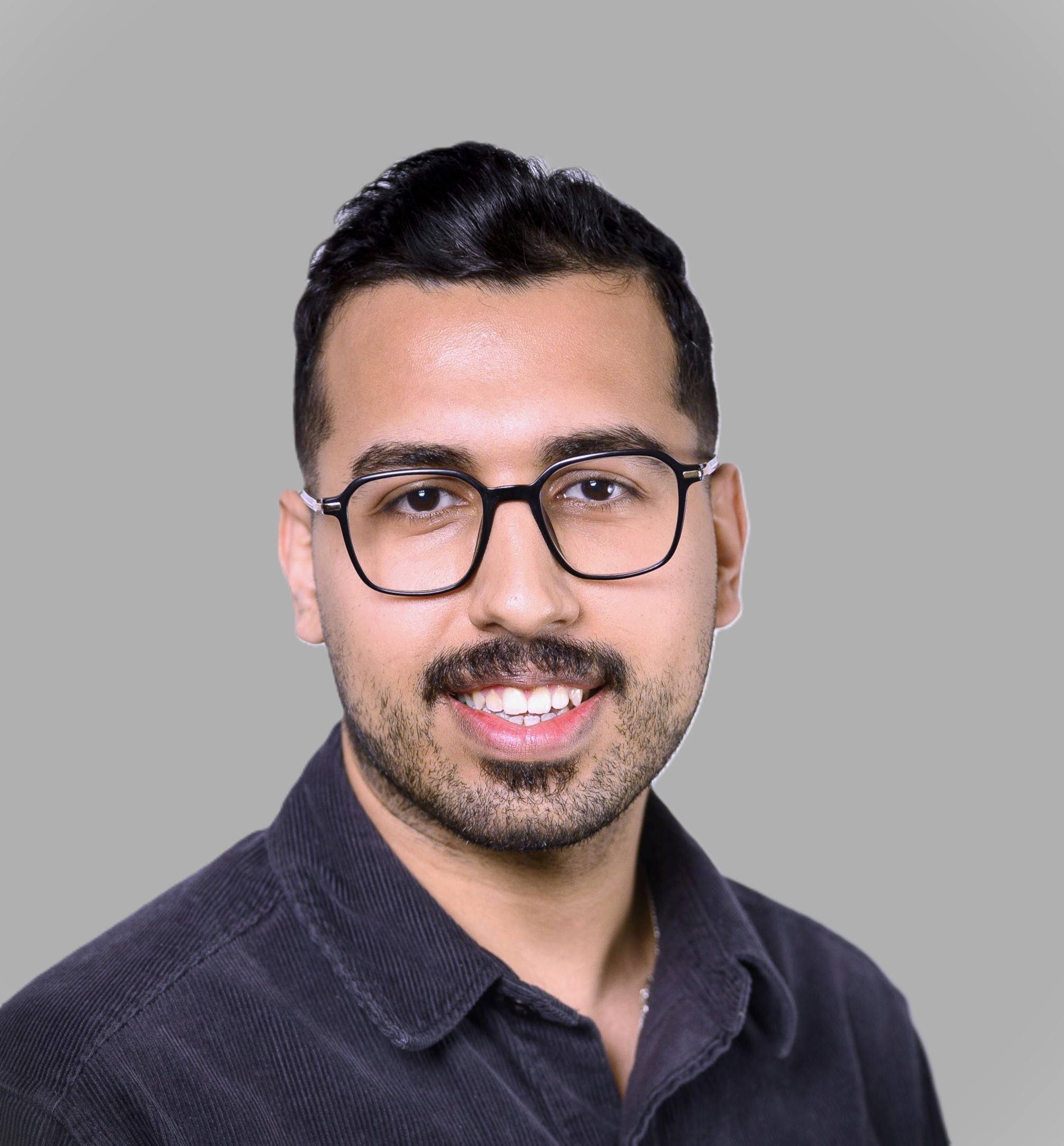}}]{S. Ali Hosseini}
received his M.Sc. degree in Systems and Control Engineering, specializing in nonlinear control (with a focus on hybrid integrator-gain systems), from Sharif University of Technology, Tehran, Iran, in 2022.

He is currently pursuing a Ph.D. in the Department of Precision and Microsystems Engineering at Delft University of Technology, Delft, The Netherlands. His research focuses on addressing industrial control challenges using nonlinear control techniques in close collaboration with ASMPT, Beuningen, The Netherlands. His research interests include precision motion control, nonlinear control systems (such as reset and hybrid systems), and mechatronic system design.
\end{IEEEbiography}
\begin{IEEEbiography}[{\includegraphics[width=1in,height=1.25in,clip,keepaspectratio]{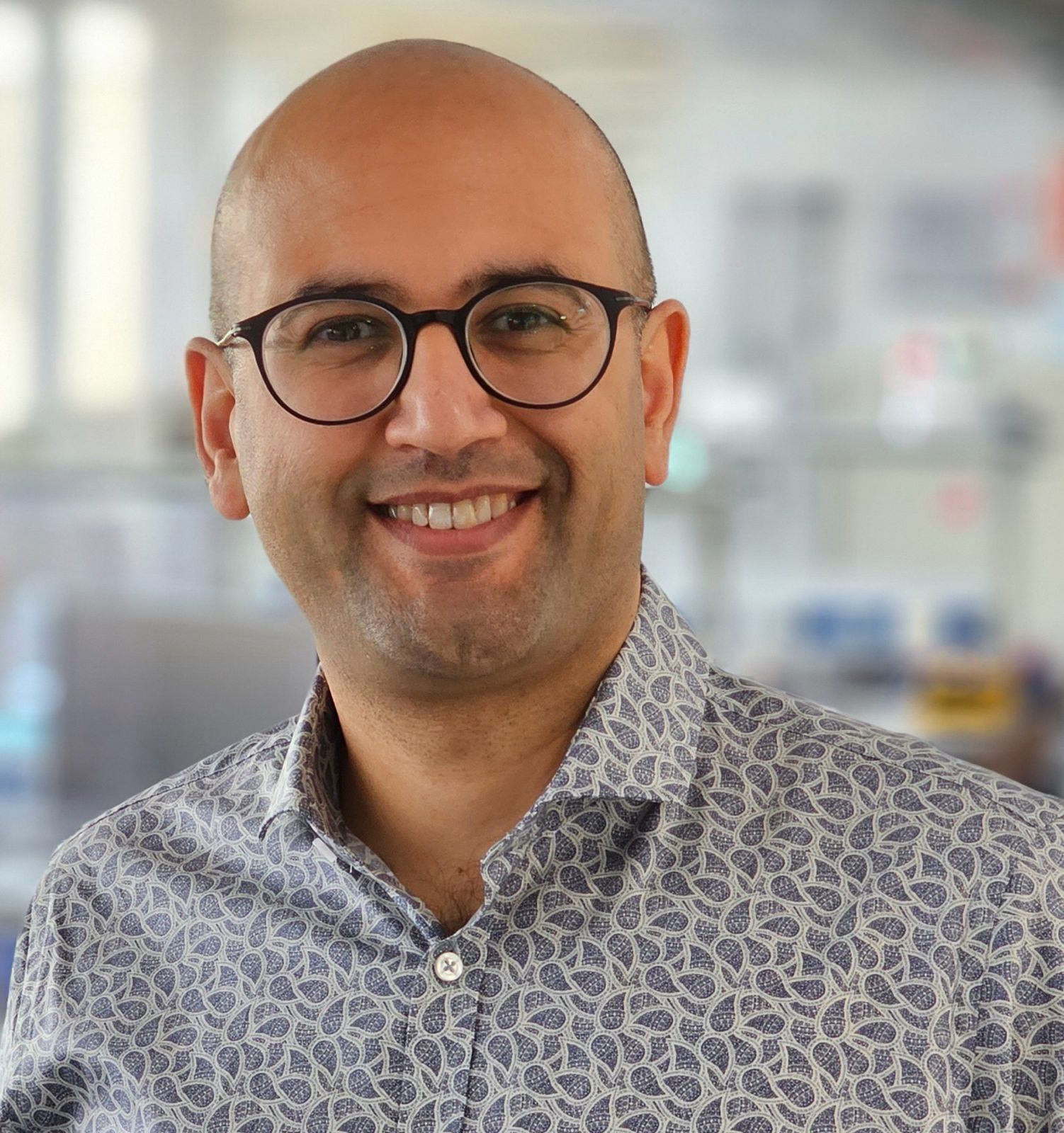}}]{S. Hassan HosseinNia} a (Senior Member, IEEE) received the Ph.D. degree (Hons.) (cum laude) in
electrical engineering specializing in automatic control: application in mechatronics from the University of Extremadura, Badajoz, Spain, in 2013. He
has an industrial background, having worked with
ABB, Sweden. Since October 2014, he has been
appointed as a Faculty Member with the Department
of Precision and Microsystems Engineering, Delft
University of Technology, Delft, The Netherlands.
He has co-authored numerous articles in respected
journals, conference proceedings, and book chapters. His main research interests include precision mechatronic system design, precision motion control,
and mechatronic systems with distributed actuation and sensing.
Dr. HosseinNia served as the General Chair of the 7th IEEE International
Conference on Control, Mechatronics, and Automation (ICCMA 2019).
Currently, he is an editorial board member of “Fractional Calculus and Applied
Analysis,” “Frontiers in Control Engineering,” and “International Journal of
Advanced Robotic Systems (SAGE).”
\end{IEEEbiography}



\end{document}